\documentclass[a4paper, 11pt]{article}
\usepackage[pdftex, a4paper]{geometry}
\usepackage{amsmath,amsfonts,amsthm,amssymb,amsbsy,amsopn,amstext}
\usepackage{authblk}
\usepackage{algorithm}
\usepackage{array}
\usepackage{textcomp}
\usepackage{stfloats}
\usepackage{url}
\usepackage{verbatim}
\usepackage{graphicx}
\usepackage{cite}
\usepackage{authblk}
\usepackage{latexsym,amssymb,amsbsy,amsopn,amstext,amsthm}
\usepackage{epsfig,epstopdf}
\usepackage[mathscr]{eucal}
\usepackage{color,multicol,makeidx}
\usepackage{bm,mathrsfs}
\usepackage{float}
\usepackage{mathrsfs}
\usepackage{tabu}
\usepackage{enumerate}
\usepackage{indentfirst} 
\usepackage{verbatim}
\usepackage{algorithmicx}
\usepackage{algpseudocode}
\usepackage{tabularx}
\usepackage[labelformat=simple]{subcaption}

\def\eq{\displaystyle\stackrel\triangle=}
\def\xra{\xrightarrow}
\def\ra{\rightarrow}
\def\logten{\log_{10}}

\newtheorem{assum}{Assumption}
\newtheorem{thm}{Theorem}
\newtheorem{prop}{Proposition}
\newtheorem{lem}{Lemma}
\newtheorem{rem}{Remark}

\newtheorem{defn}{Definition}

\begin{document}

\title{\bf RSS-Based Localization: Ensuring Consistency and Asymptotic Efficiency}

\author[1]{Shenghua Hu}
\author[2]{Guangyang Zeng}
\author[1]{Wenchao Xue}
\author[1]{Haitao Fang}
\author[2]{Junfeng Wu}
\author[1]{Biqiang Mu}

\affil[1]{State Key Laboratory of Mathematical Sciences, Academy of Mathematics and Systems Science, Chinese Academy of Sciences, Beijing 100190, China\\
\texttt{hushenghua17@mails.ucas.ac.cn; wenchaoxue@amss.ac.cn; htfang@iss.ac.cn; bqmu@amss.ac.cn}}
\affil[2]{School of Data Science, Chinese University of Hong Kong, Shenzhen, Shenzhen, China\\
\texttt{zengguangyang@cuhk.edu.cn; junfengwu@cuhk.edu.cn}}

\maketitle

\begin{abstract}
We study the problem of signal source localization using received signal strength measurements. We begin by presenting  verifiable geometric conditions for sensor deployment that ensure the model's asymptotic localizability. Then we establish the consistency and asymptotic efficiency of the maximum likelihood (ML) estimator. However, computing the ML estimator is challenging due to its reliance on solving a non-convex optimization problem.
To overcome this, we propose a two-step estimator that retains the same asymptotic properties as the ML estimator while offering low computational complexity—linear in the number of measurements. The main challenge lies in obtaining a consistent estimator in the first step. To address this, we construct two linear least-squares estimation problems by applying algebraic transformations to the nonlinear measurement model, leading to closed-form solutions.
In the second step, we perform a single Gauss-Newton iteration using the consistent estimator from the first step as the initialization, achieving the same asymptotic efficiency as the ML estimator. Finally, simulation results validate the theoretical property and practical effectiveness of the proposed two-step estimator.
\end{abstract}

\begin{keywords}
Received signal strength, Source localization  Asymptotic localizability, Consistency,  Asymptotic efficiency, Two-step
estimators
\end{keywords}

\section{Introduction}

Signal source localization is the process of determining the location of signal source  from measurements collected by sensors. It plays a crucial role in a wide range of applications, including radar, sonar,  wireless networks, cognitive radio networks, multimedia systems \cite{Niu2018,Li2002,Gezici2005}, where accurate source location determination is essential. The methods for signal source localization depend on the signal's characteristics, the surrounding environment, and the sensors used. Generally, these methods are categorized based on the types of measurements they utilize, such as the time of arrival (TOA) \cite{Soares2015, Han2016, Zeng2022}, time difference of arrival (TDOA) \cite{Ho2012, Sun2019, Zeng2024}, angle of arrival (AOA) \cite{ Shao2014, Sun2020,Zhou2024}, and received signal strength (RSS) \cite{So2011, Angjelichinoski2015, Liu2016,Zou2021,Aubry2023,Li2024TVT}. Among these, RSS-based localization stands out due to its simple implementation, low cost, and minimal communication requirements \cite{Liu2016, Wang2019, Shi2020}. RSS measurements can be easily obtained by most wireless devices \cite{Angjelichinoski2015}. Moreover, compared to TOA localization, RSS-based localization does not require clock time synchronization \cite{So2011}. These advantages make RSS-based localization methods particularly attractive for applications such as indoor geolocation where GPS is ineffective, surveillance and military operations \cite{Ying2019, Hu2017tsp}. Therefore, this paper focuses on signal source localization using RSS measurements.

  RSS-based localization is essentially equivalent to TOA localization in the absence of noise since the RSS measurement is a bijection of range measurement. However, in practical scenarios, RSS measurements are often corrupted by noise, requiring the source coordinates to be deduced from noisy data. Therefore, RSS  source localization is a  typical parameter estimation problem. The maximum likelihood (ML) and least squares (LS) estimators  are the most commonly utilized estimators for source localization. 
  In most cases, the measurement   noise is  assumed to be independent and identically distributed (i.i.d.) Gaussian random  variables, which makes the two estimators equivalent. Even for the simple  Gaussian noise, it is challenging to solve the ML or LS estimator due to the associated non-convex optimization problems.

For RSS-based localization, there are primarily two kinds of methods for solving the ML and LS estimators: the iteration-based  and relaxation-based methods. 
The iteration-based methods typically employ iterative optimization algorithms, such as gradient descent, Gauss-Newton (GN), and Levenberg-Marquardt methods, to search for the global optimum of the ML and LS estimation problems \cite{Li2006, Cheng2009}. 
The gradient descent method with a weighted least squares (WLS) initialization is applied to localize the true source in \cite{Tarrio2011}. 
  In  \cite{Coluccia2014}, the source location is iteratively determined using ranges estimated via an empirical Bayesian method, which assesses the distances between sensors and the signal source. 
The relaxation-based methods utilize measurement model linearization, loss function approximation, and constraint relaxation to approximately solve the ML and LS estimation problems. Under the small noise assumption, the linearization of the RSS-based measurement model introduces an auxiliary parameter and a norm constraint associated with the true source, leading to constrained linear LS or WLS estimation problems \cite{Wang2019,Shi2020,Vaghefi2013,So2011,Salman2014,Larsson2025}. The constrained linear LS problems are typically solved using semidefinite programming (SDP) \cite{Vaghefi2013} or second-order cone programming \cite{Tomic2015}, while the linear WLS problems are addressed using approximated covariance matrices and norm constraints \cite{So2011},   approximated ranges from the LS solution \cite{Salman2014}, or an eigenvalue method aided by a customized transformation function \cite{Larsson2025}.

Existing approaches to RSS-based localization face  limitations. The iterative optimization methods  suffer from high computational complexity and often converge to a stationary point rather than the global solution. The relaxation-based methods typically introduce asymptotic bias stemming from small-noise approximations. Particularly, this bias increases as noise levels increase. Moreover, the SDP-based relaxed problems are computationally expensive to solve. Overall, existing methods lack consistency and asymptotic efficiency\footnote{``Consistency" refers to the estimator converging to the true value as the number of measurements increases, while ``asymptotic efficiency" means that, as the number of measurements grows, the mean squared error (MSE) approaches the Cramér-Rao lower bound (CRLB), the theoretical minimum MSE for any unbiased estimator.}.
Furthermore, the fundamental question of asymptotic localizability\footnote{``Asymptotic localizability" describes the property that the source location can be uniquely determined asymptotically using available measurements. This concept is rigorously formalized in Definition \ref{def_asy_loc}.} for RSS-based localization has yet to be sufficiently explored in the literature.
To address these gaps, this paper makes two key contributions:
  (i) we establish the asymptotic localizability of RSS-based localization; (ii) we develop an estimator that is consistent and asymptotically efficient with low computational complexity.

To establish asymptotic localizability, we develop verifiable geometric conditions for sensor deployment, such as non-cohyperplanarity and non-cohypersphericity. These conditions ensure that the sensor network can uniquely determine the true source asymptotically.
To obtain a consistent and asymptotically efficient estimator for RSS-based localization, we first develop the ML estimator and rigorously prove its consistency and asymptotic efficiency under suitable conditions. However, the non-convex nature of the ML estimator poses challenges for accurate and efficient computation.
To address the challenges, we introduce a two-step estimator: (i) deriving a consistent estimator of the source location; (ii) running a single GN iteration using this consistent estimator  as the initial value.
A key advantage of the two-step estimator is its guaranteed asymptotic efficiency when the initial consistent estimator
has a convergence rate of $O_p(1/\sqrt{n})$ \cite{Lehmann1998}, where $n$ is the number of measurements.
The two-step estimator has been successfully applied in parameter estimation of nonlinear rational models\cite{Mu2017}, and TOA and TDOA-based localization \cite{Zeng2022,Zeng2024}.
 Therefore, the key to obtaining the asymptotically efficient estimator of the RSS-based localization lies in deriving a  consistent estimator with the convergence rate of $O_p(1/\sqrt{n})$.
 This is accomplished through the LS estimators associated with RSS-based localization.
 
 Deriving a consistent estimator for RSS-based localization is challenging due to the composite structure of logarithmic and norm functions in the measurement model, as shown in \eqref{RSS}.
For known noise variance, through appropriate model transformation, we obtain a linear regression model where the parameter vector comprises the true source and its inner product. The resulting unconstrained linear LS estimator    provides a consistent estimator for the true source with the rate $O_p(1/\sqrt{n})$, which removes the norm constraint on the parameter vector and avoids the need to solve the constrained non-convex optimization problems as in the relaxation-based methods.
For unknown noise variance, the method remains similar but requires augmenting the parameter vector with an additional constant and introducing an extra assumption on sensor deployment geometry to obtain a consistent estimator for the true source with the rate $O_p(1/\sqrt{n})$.
The computational complexity of our two-step estimator is dominated by the LS estimation and a single GN iteration, resulting in linear scaling with the number of measurements and ensuring high efficiency.

The proposed two-step estimator achieves asymptotic efficiency and computational efficiency, ensuring optimal estimation accuracy with low computational resource for large-scale RSS measurements.
Localization based on massive RSS measurements is widely applicable in practical scenarios, particularly due to the simplicity and low cost of RSS devices, which enable scalable deployment. The growing prevalence of wireless devices and networks has further increased the availability of RSS observation points \cite{Lee2011book}. Additionally, massive multiple-input multiple-output (MIMO) technology—a key enabler of 5G \cite{Prasad2017}—enhances the large-scale collection of RSS measurements for user positioning \cite{Prasad2018}.

The rest of the paper is organized as follows. Section \ref{sec:rss} formulates RSS-based localization problem.
Section \ref{sec:ML} develops sufficient geometric conditions for guaranteeing  the asymptotic localizability.
Section \ref{secmlts} introduces  the ML estimator and the two-step estimator. Section \ref{sqrtncon} represents the $\sqrt{n}$-consistent estimator for the RSS-based localization.   Section \ref{sec:algo} presents the algorithm for the two-step estimator with the analysis of computational complexity. Section \ref{sec:sim} demonstrates the estimation accuracy and computational efficiency of the estimator using extensive Monte-Carlo simulations. Finally, Section \ref{sec:con} concludes the paper. The proofs of all the theorems, propositions and lemmas in the main body of the paper can be found in 
the Appendix A.

\textbf{Notation.} 
In this paper, \(\mathbb{E}\) and \(\mathbb{V}\) denote the expectation and variance with respect to the distribution of the noise, respectively, unless otherwise specified. The superscript \((\cdot)^0\) indicates the true or noise-free value of a given quantity. For a sequence of random variables \(X_n\), the notation \(X_n = O_p(c_n)\) means that \(X_n / c_n\) is bounded in probability, while \(X_n = o_p(c_n)\) indicates that \(X_n / c_n\) converges to zero in probability. Lastly, \(\nabla\) and \(\nabla^2\) denote the first and second order differential operators, respectively.

\section{Problem formulation for RSS-based  localization}
\label{sec:rss}
In this section, we propose the problem formulation for RSS-based source localization.

Suppose there are \(n\) sensors scattered in an \(m\)-dimensional space, where $m=2$ or $m=3$. Let \(p^0 \in \mathbb{R}^m\) denote the coordinates of the signal source, which are unknown and need to be estimated using the RSS measurements. Let \(p_i \in \mathbb{R}^m\) represent the coordinates of sensor \(i\), \(i = 1, \dots, n\), which are known exactly. 
Let $$d_i \eq \|p_i - p^0\|$$ represent the distance between sensor \(i\) and the source. In the absence of disturbance, the noise-free signal strength received by sensor \(i\) is given by:
\begin{equation}
\label{Friis}
    P_i^0 = P_0\big/d_i^\alpha, \quad i=1,...,n,
\end{equation}
where \( P_0 \) is a positive constant associated with factors affecting signal strength, and \( \alpha \) is the path loss exponent, typically ranging from 1 to 5 depending on the propagation environment, with \( \alpha = 2 \) corresponding to free-space conditions \cite{So2011}.  The model in \eqref{Friis} is also known as the Friis free space equation \cite{Rappaport2024}. The observed RSS in decibel scale is given by
\begin{equation}
\label{RSS_raw}
    10\log_{10}(P_i) = 10\log_{10}(P_0) - 10\alpha\log_{10}(d_i)+ \varepsilon_i,~ i=1,...,n,
\end{equation}
where $\log_{10}(\cdot)$ denotes the logarithm of a positive number with base 10 and $\varepsilon_i$ is the measurement noise.

\begin{rem}
\label{rem1}
Different formulations of the RSS model can be found in the literature \cite{Vaghefi2013,Salman2014,Larsson2025}, but they all share the same underlying Friis model \eqref{Friis}. 
\end{rem}

For the model \eqref{RSS_raw}, we make the following assumption.
\begin{assum}
\label{assum_noise}
\begin{enumerate}[(i)]
    \item The constants $\alpha$ and $P_0$ are known exactly.
    \item The noises $\{\varepsilon_i\}_{i=1}^n$ are i.i.d. Gaussian random variables with mean zero and finite variance $\sigma^2$.
\end{enumerate}
\end{assum}
\begin{rem}
 The constants $\alpha$ and $P_0$ are usually supposed to be known in prior through a testing and calibration campaign  \cite{Patwari2005,Tarrio2008,Coluccia2010,So2011}.
The measurement model \eqref{RSS_raw} has been validated to follow the log-normal shadowing nature, i.e. the noises $\{\varepsilon_i\}_{i=1}^n$ are i.i.d. additive Gaussian random variables \cite{Coulson1998,Patwari2005}.
\end{rem}

Under Assumption \ref{assum_noise}, define the equivalent RSS measurements and equivalent measurement noises as  
\[ y_i \eq -\frac{\log_{10}(P_i) - \log_{10}(P_0)}{\alpha}, ~ \omega_i \eq -\frac{\varepsilon_i}{10\alpha}, ~  i = 1, \dots, n. \]  
This yields the following equivalent RSS measurement model:
\begin{equation}
    \label{RSS}
    y_i = \logten(d_i)+\omega_i,\quad i=1,...,n,
\end{equation}
where $\{y_i\}_{i=1}^n$ are available equivalent measurements and $\{\omega_i\}_{i=1}^n$ are i.i.d. Gaussian random variables with mean zero and finite variance $\sigma^2/(100\alpha^2)$.
 In what follows, we focus on the equivalent model \eqref{RSS}. 
Therefore, RSS-based localization problem is to determine the true source corrdinates $p^0$ using the sensor corrdinates and the corresponding RSS measurements $\{p_i,y_i,i=1,...,n\}$.

\section{Asymptotic localizability of RSS-based localization}
\label{sec:ML}
In this section, we present sufficient conditions on sensor geometric deployment to ensure the asymptotic localizability of RSS-based localization.

\begin{assum}
	\label{assum_coordinates}
The source \(p^0\) lies within a bounded set \(\mathcal{P}^0\), and the sensors \(p_i\), for \(i = 1, \dots, n\), belong to a bounded set \(\mathcal{P}\), independent of \(n\). Moreover, the sets are disjoint, i.e., \(\mathcal{P}^0 \cap \mathcal{P} = \emptyset\).
\end{assum}
In practice, Assumption \ref{assum_coordinates} is readily fulfilled because sensors cannot be positioned arbitrarily far from the source. Additionally, the set \(\mathcal{P}^0\) can always be defined as a compact region surrounding the source, without needing prior knowledge of the source's precise location.

\begin{defn}
\label{def:empirical_dist}
Let \( x_1, \ldots, x_n \) be a sequence in a measurable space \((\Omega, \mathcal{F})\). The empirical distribution \( P_n \) is the discrete probability measure on \((\Omega, \mathcal{F})\) defined by
$
P_n \eq\frac{1}{n} \sum_{i=1}^n \delta_{x_i},
$
where \(\delta_{x_i}\) is the Dirac measure at \( x_i \), i.e., \(\delta_{x_i}(A) = \mathbf{1}_{\{x_i \in A\}}\) for \( A \in \mathcal{F} \).
\end{defn}

\begin{assum}
	\label{assum_cohyperplane}
	\begin{enumerate}[(i)]
	    \item The empirical distribution function $F_n$ of the sensor sequence $p_1, p_2,...$ converges to a distribution function $F_\mu$ and the probability measure induced by the distribution \(F_\mu\)  is denoted by \(\mu\);
	    \item For any positive integer \(n\), the sensors \(p_1, \dots, p_n\) do not lie on a line when \(m = 2\), nor on a plane when \(m = 3\).  
Moreover, there does not exist any subset $\mathcal{P}'$ of \(\mathcal{P}\) with \(\mu(\mathcal{P}') = 1\) such that \(\mathcal{P}'\) lies entirely on a line for \(m = 2\), or on a plane for \(m = 3\).
	\end{enumerate}
\end{assum}
 
Denote by  
$
f_i(p) \eq \log_{10}(\|p_i - p\|)  
$ 
the  predictive RSS  for model \eqref{RSS} for sensor \(i\), where \(i = 1, \dots, n\)
and define
\begin{align}
h_n(p) \eq   \frac{1}{n}\sum_{i=1}^n(f_i(p)-f_i(p^0))^2.
\end{align}
We now present the definition of asymptotic localizability for RSS-based localization.
\begin{defn}
\label{def_asy_loc}
For the model \eqref{RSS}, the true signal source $p^0$ is called asymptotically localizable if   $h_n(p)$ has a limit function and the limit function, denoted by   $h(p)\eq \lim_{n \to \infty}h_n(p)$,  has a unique minimum at $p=p^0$.
\end{defn}

 We have the following results on the asymptotic localizability for RSS-based localization.

\begin{thm}
\label{lem_uniq}
    Under Assumptions \ref{assum_coordinates}-\ref{assum_cohyperplane}, the signal source is asymptotically   localizable.
\end{thm}

\section{Maximum likelihood estimator and two-step estimator}
\label{secmlts}
This section first derives the ML estimator for the RSS-based localization model defined in \eqref{RSS} and rigorously proves its consistency and asymptotic efficiency. Subsequently, we propose a two-step estimator to  asymptotically aproximate the global  maximum of the non-convex optimization problem associated with the ML estimator.
 
\subsection{The maximum likelihood estimator}

Based on Assumption \ref{assum_noise}(ii), the log-likelihood function of the model \eqref{RSS} is given by
\begin{equation*}
    \ell_n(p) \eq -n\ln\left(\frac{\sqrt{2\pi}\sigma}{10\alpha}\right) - \frac{50\alpha^2}{\sigma^2}\sum_{i=1}^n\left( y_i - \logten(\|p_i-p\|) \right)^2,
\end{equation*}
where $\ln(\cdot)$ denotes the natural logarithm of a positive number.
The associated ML estimation problem is
\begin{equation}
    \label{ML}
    {\rm \textbf{ML}}: \quad \min_{p \in \mathbb{R}^m} \frac{1}{n}\sum_{i=1}^n\left( y_i - \logten(\|p_i-p\|) \right)^2,
\end{equation}
  which is equivalent to maximize $\ell_n(p)$. The ML estimator, denoted by $\widehat{p}_n^{\rm ML}$, is the point that maximizes $\ell_n(p)$.

Let $\nabla f_i(p)$ be the gradient vector of $f_i(p)$ with respect to $p$, i.e.,\begin{align*}
    \nabla f_{i}(p) \eq \frac{\partial f_i(p)}{\partial p}= \frac{(p-p_i)}{\|p_i-p\|^2\ln(10)}.
\end{align*} For proving the asymptotic efficiency of the ML estimator \eqref{ML}, we  first  derive its asymptotic variance, which is given in the following lemma.
\begin{lem}
    \label{lem_converge_M}
Under that Assumptions \ref{assum_coordinates}-\ref{assum_cohyperplane}, we have 
    \begin{enumerate}[(i)]
        \item  The  average $ \frac{1}{n}\sum_{i=1}^n\nabla f_i(p)\nabla f_i(p)^T$ converges uniformly for $p \in \mathcal{P}^0$  as $n \to \infty$.
        \item Moreover, the limit matrix $M^0\eq$ $\lim_{n\to \infty}\frac{1}{n}\sum_{i=1}^n\nabla f_i(p^0) \nabla f_i(p^0)^T$ is nonsingular.
    \end{enumerate}
 
\end{lem}

Now, we present the asymptotic efficiency of the ML estimator. 
The proof is straightforward by checking the conditions in \cite[Theorem 3]{Jennrich1969}, and we omit the proof.

\begin{prop}
\label{thm_consis_asymnormal}
    Under Assumptions \ref{assum_noise}-\ref{assum_cohyperplane}, we have $\widehat{p}_n^{\rm ML} \to p^0$ almost surely as $n \to \infty$ with the asymptotic normality
    \begin{equation}
        \label{asy_normal}
        \sqrt{n}(\widehat{p}_n^{\rm ML}-p^0) \to \mathcal{N}\left(0,\frac{\sigma^2}{100\alpha^2}\left(M^0\right)^{-1}\right)   {\rm as} ~ n\to \infty.
    \end{equation}
\end{prop}

The limit matrix $M^0$ is tightly related to the Fisher information matrix $F$ of model \eqref{RSS}. To see this, firstly we have
\begin{equation*}
    \frac{\partial \ell_n(p^0)}{\partial p^0} = \frac{100\alpha^2}{\sigma^2}\sum_{i=1}^n \omega_i \frac{p^0-p_i}{\|p_i-p^0\|^2\ln(10)}.
\end{equation*}
Then we obtain the Fisher information matrix
\begin{align*}
    F =& \mathbb{E}\left[ \frac{\partial \ell_n(p^0)}{\partial p^0} \left(\frac{\partial \ell_n(p^0)}{\partial p^0}\right)^T \right]\\
    =& \frac{100\alpha^2}{\sigma^2} \sum_{i=1}^n\frac{1}{\|p_i-p^0\|^4\ln^2(10)}(p^0-p_i)(p^0-p_i)^T.
\end{align*}
This means $\lim_{n\to \infty}nF^{-1} = \frac{\sigma^2}{100\alpha^2} \left(M^0\right)^{-1}$, which implies that the ML estimator $\widehat{p}_n^{\rm ML}$ is asymptotically efficient. Notice that the CRLB is given by the trace of the inverse of the Fisher
 information matrix, i.e., CRLB $= {\rm tr}(F^{-1})$. From \eqref{asy_normal}, we have $\widehat{p}_n^{\rm ML}$ converges to the true value $p^0$ with the asymptotic covariance of $\frac{\sigma^2}{100\alpha^2}\left(M^0\right)^{-1}/n$.


\subsection{The two-step estimator}
\label{sec2c}
The ML estimator owns consistency and asymptotic efficiency, but it is difficult to obtain the  global solution to the ML problem~\eqref{ML} due to its non-convexity. In this subsection, we introduce a two-step estimation scheme, which can realize the same asymptotic property that the ML estimator possesses.

Firstly, we prove that the  objective function of the ML problem converges to a function that is convex in a small neighborhood around $p^0$, which forms the feasibility of the two-step scheme. 

\begin{prop}
    \label{thm_converge_ML_objfunc}
Under Assumptions \ref{assum_noise}-\ref{assum_cohyperplane}, $\ell_n(p)/n$ converges uniformly to 
\begin{equation*}
    \ell(p)\eq -\ln\left(\frac{\sqrt{2\pi}\sigma}{10\alpha}\right) -\frac{50\alpha^2}{\sigma^2} h(p) -\frac12
\end{equation*}
on $\mathcal{P}^0$  as $n \to \infty$. In addition, $\nabla^2(-\ell(p^0)) = 100\alpha^2 M^0/\sigma^2$.
\end{prop}

Proposition \ref{thm_converge_ML_objfunc} indicates that $-\ell_n(p)/n$ is a convex function in a small neighborhood around the global minimum $p^0$ when $n$ is large. Therefore, local iterative methods—such as GN iterations—can be used to find the global minimum of the non-convex problem, provided that the initial value lies within this region of attraction.
Notably, a consistent estimator can approach the true value arbitrarily closely as 
$n$ increases.
Based on this property, the two-step estimator is formally defined as follows \cite{Gourieroux1995,Mu2017}:

\textbf{Step 1.} Derive a  consistent estimator $\widehat{p}_n$ for the source's coordinates $p^0$.

\textbf{Step 2.} Run the GN iteration algorithm with this consistent estimator $\widehat{p}_n$ as its initial value.

In Step 2, the GN iteration algorithm for the ML problem \eqref{ML} has the following form:
\begin{align} \label{gn}
\widehat{p}_n(k+1) &= \widehat{p}_n(k) +  \left(J(k)^T J(k)\right)^{-1}J(k)^T(y - f(k)),\!\!
\end{align}
where 
\begingroup
\allowdisplaybreaks
\begin{align*}
J(k)& = \left[\nabla f_1(\widehat{p}_n(k))^T, ~ \cdots, ~ \nabla f_n(\widehat{p}_n(k))^T\right]^T,\\
 f(k)& = [f_1(\widehat{p}_n(k)), ~ \cdots, ~ f_n(\widehat{p}_n(k))]^T,\\
    y &=  [y_1,~\cdots, ~y_n]^T,~~
    \widehat{p}_n(0) = \widehat{p}_n.
\end{align*}
\endgroup
 The two-step scheme described above has the appealing property that a single step of the GN iteration is sufficient to theoretically ensure that the resulting two-step estimator inherits the same asymptotic properties as the ML estimator \cite{Mu2017,Lehmann1998}. This result is summarized in the following lemma.

\begin{lem}{\cite[Theorem 2]{Mu2017}\cite[Chapter 6, Theorem 4.3]{Lehmann1998}}
\label{lem_onestepGN}
    Suppose that $\widehat{p}_n$ is a $\sqrt{n}$-consistent estimator of $p^0$, i.e., $\widehat{p}_n-p^0 = O_p(1/\sqrt{n})$. Denote the one-step GN iteration of $\widehat{p}_n$ by $\widehat{p}_n^{\rm GN}\eq\widehat{p}_n(1)$. Then under Assumptions \ref{assum_noise}-\ref{assum_cohyperplane}, we have
$
        \widehat{p}_n^{\rm GN}-\widehat{p}_n^{\rm ML} = o_p(1/\sqrt{n}).
$
    This indicates that $\widehat{p}_n^{\rm GN}$ has the same asymptotic property that $\widehat{p}_n^{\rm ML}$ possesses.
\end{lem}
 
Lemma \ref{lem_onestepGN} shows that the estimator $\widehat{p}_n^{\rm GN}$ is both consistent and asymptotically efficient. Thus, the success of the two-step estimation scheme hinges on obtaining a $\sqrt{n}$-consistent estimator in the first step. In the next  section, we derive such an estimator.

\section{Deriving $\sqrt{n}$-consistent estimator for RSS-based localization}
\label{sqrtncon}
In this section, we focus on deriving  $\sqrt{n}$-consistent estimators for RSS-based localization, which is achieved by model transformation and least squares estimators.
\subsection{Model transformation}
By multiplying both sides of the model in equation \eqref{RSS} by 2 and then applying the base-10 exponential function, we obtain an equivalent form of the model in equation \eqref{RSS}
\begin{equation}
\label{eq8}
    10^{2y_i} = d_i^2 10^{2\omega_i}, \quad i=1,...,n.
\end{equation}
Note that \( \{ \omega_i \}_{i=1}^n \) are i.i.d. Gaussian random variables. Therefore, \( \{ 10^{2\omega_i} \}_{i=1}^n \) are i.i.d. lognormal random variables \cite{Johnson1994}, and their mean and variance are provided in the following lemma.
\begin{lem}
\label{lem_eta}
Under Assumption \ref{assum_noise}, 
the  sequence $\{10^{2\omega_i}\}_{i=1}^n$ has  mean
$b \eq \mathbb{E}(10^{2\omega_i})= e^{(\ln 10)^2\sigma^2/(50\alpha^2)} >1$
and variance  
    $ \mathbb{V} (10^{2\omega_i}) = b^2(b^2-1) < \infty$.
\end{lem}

 By Lemma \ref{lem_eta}, we have $\{\eta_i \eq 10^{2\omega_i} - b,~i=1,...,n\}$ is an i.i.d. random variable sequence with mean zero and variance $b^2(b^2-1)$. Note that  $d_i^2 = \|p^0\|^2 - 2p_i^Tp^0 + \|p_i\|^2,~i=1,...,n$.
 We reach the desired model from \eqref{eq8}
\begin{align}
\nonumber
10^{2y_i} &= d_i^2 b + d_i^2 (10^{2\omega_i} -b)\\
&=(\|p^0\|^2 - 2p_i^Tp^0 + \|p_i\|^2)b+d_i^2 \eta_i.\label{lrm1}
\end{align}

\subsection{Least squares estimator with known noise variance}

When the noise variance \(\sigma^2\) is known, 
we  define $\theta^0 \eq [(p^0)^T,~\|p^0\|^2]^T$ and obtain from \eqref{lrm1}
\begin{equation}
\label{para_model_1}
    10^{2y_i}-b\|p_i\|^2 = b[-2p_i^T,~1]\theta^0 + d_i^2\eta_i,~ i=1,...,n.
\end{equation}
We write \eqref{para_model_1} as the compact form 
\begin{equation}
    Y = X\theta^0+V, \label{lrm}
\end{equation}
where
\begin{equation*}
    Y = \begin{bmatrix} 10^{2y_1}-b\|p_1\|^2 \\ \vdots \\ 10^{2y_n}-b\|p_n\|^2\end{bmatrix},~ X=  b\begin{bmatrix} -2p_1^T & 1\\ \vdots & \vdots \\ -2p_n^T & 1 \end{bmatrix}, ~ V = \begin{bmatrix} d_1^2\eta_1 \\ \vdots \\ d_n^2\eta_n\end{bmatrix}.
\end{equation*}
In this case, the constant \(b\) is determined according to Lemma \ref{lem_eta}. Given that \(\{d_i^2 \eta_i\}_{i=1}^n\) is an independent random variable sequence with zero mean and finite variance, the LS estimator of the model in \eqref{lrm} is consistent and \(\sqrt{n}\)-consistent, provided that the Gram matrix \(X^TX/n\) is nonsingular. This is established in the following proposition.
\begin{prop}
\label{prop_invertgram_1}
    Under Assumptions \ref{assum_noise}-\ref{assum_cohyperplane} , the Gram matrix $X^TX/n$ is nonsingular for any finite $n$. Moreover, the limit $\lim_{n \to \infty}X^TX/n$ exists and is nonsingular.
\end{prop}

Based on Proposition \ref{prop_invertgram_1}, we define the LS estimator for the model \eqref{lrm} as follows:
\begin{equation}
\label{hattheta_knownvar}
    \widehat{\theta}_n = (X^TX)^{-1}X^TY.
\end{equation}
Thus, the true source is estimated by the first $m$ elements of $\widehat{\theta}_n$, defined by
\begin{equation}
\label{hatp_knownvar}
    \widehat{p}_n = [\widehat{\theta}_n]_{1:m}.
\end{equation}
\begin{rem}
We disregard the constraint  $\|[\theta^0]_{1:m}\|^2 = [\theta^0]_{m+1}$ in the recovery of the source to avoid solving a non-convex problem, which is typically computationally demanding. This is because our goal in this step is to derive a $\sqrt{n}$-consistent estimator, and the proposed estimator $\widehat{p}_n$ is sufficient for this purpose.
\end{rem}
We demonstrate that the proposed estimator \eqref{hatp_knownvar} for the source is $\sqrt{n}$-consistent, as stated in the following theorem.
\begin{thm}
\label{thm_sqrtn_esti_1}
    Under Assumptions \ref{assum_noise}-\ref{assum_cohyperplane} , it holds that 
    \begin{equation}
        \widehat{p}_n-p^0 = O_p(1/\sqrt{n}).
    \end{equation}
\end{thm}
 
\subsection{Least squares estimator with unknown noise variance}

When the noise variance $\sigma^2$ is unknown, both $b$ and   $X$ in the model \eqref{lrm} are not available.
In this case, we resort to define $\beta^0 \eq b[(p^0)^T,~\|p^0\|^2,~1]^T$ with the constant $b$ being positive by Lemma \ref{lem_eta} and formulate the model \eqref{lrm1} as 
\begin{equation}
\label{para_model_2}
    10^{2y_i} = [-2p_i^T,~1,~\|p_i\|^2]\beta^0+d_i^2\eta_i,~ i=1,...,n.
\end{equation}
Then we obtain the linear regression model
\begin{equation}
    \Gamma = \Phi\beta^0 + V,\label{lrm2}
\end{equation}
where
\begin{equation*}
    \Gamma = \begin{bmatrix} 10^{2y_1}\\ \vdots \\ 10^{2y_n} \end{bmatrix},~\Phi = \begin{bmatrix} -2p_1^T & 1 & \|p_1\|^2 \\ \vdots & \vdots & \vdots \\ -2p_n^T & 1 & \|p_n\|^2 \end{bmatrix}, ~ V = \begin{bmatrix} d_1^2\eta_1 \\ \vdots \\ d_n^2\eta_n\end{bmatrix}.
\end{equation*}
The model \eqref{lrm2} introduces an additional covariate, $\|p_i\|^2$, compared to the model \eqref{lrm}. To recover the sources from the LS estimator of the model \eqref{lrm2}, a stronger assumption on the sensor deployment is required, as stated in the following assumption.
\begin{assum}
	\label{assum_cohypersphere}
 For any positive integer \(n\), the sensors \(p_1, \dots, p_n\) do not lie on a circle when \(m = 2\), nor on a sphere when \(m = 3\).  
Moreover, there does not exist any subset  $\mathcal{P}'$ of \(\mathcal{P}\) with \(\mu(\mathcal{P}') = 1\) such that \(\mathcal{P}'\) lies entirely on a circle for \(m = 2\), or on a sphere for \(m = 3\).
\end{assum}

\begin{rem}

In model \eqref{lrm1}, when the noise variance is unknown, the parameter \( b \) must be estimated alongside the regression coefficients. Notably, \( b \) is multiplied by the quadratic terms $\|p_i\|^2,~i=1,...,n$  in the model. To accommodate this, the quadratic terms \( \|p_i\|^2 \) must be explicitly included as additional covariate in the regression matrix. 
As a result, beyond Assumption~\ref{assum_cohyperplane}, an additional non-cohypersphericity condition given in Assumption \ref{assum_cohypersphere} is required to ensure the invertibility of the Gram matrix due to the presence of these quadratic terms.
 
\end{rem}

As before, we prove that the Gram matrix $\Phi^T\Phi/n$ is nonsingular, as given in   the following proposition.
\begin{prop}
\label{prop_invertgram_2}
    Under Assumptions \ref{assum_coordinates}-\ref{assum_cohypersphere} , the Gram matrix  $\Phi^T\Phi/n$ is nonsingular for any finite $n$. Further, the limit $\lim_{n \to \infty}\Phi^T\Phi/n$ exists and is nonsingular.
\end{prop}
 
 The  LS estimator for the model \eqref{lrm2} is given by
\begin{align}
\widehat{\beta}_n = (\Phi^\top \Phi)^{-1} \Phi^\top \Gamma,\label{hatbeta_unknownvar}
\end{align}
and the source is estimated by dividing the first $m$ entries of $\widehat{\beta}_n$ by the $(m+2)$-th entry, i.e.,
\begin{equation}
\label{hatp_unknownvar}
    \widehat{p}_n = [\widehat{\beta}_n]_{1:m}\big/[\widehat{\beta}_n]_{m+2}.
\end{equation}
Analogous to the case with known noise variance, we discard the constraint on the parameters $\beta$ and directly use the LS estimator $\widehat{\beta}_n$ to obain the $\sqrt{n}$-consistent estimator \eqref{hatp_unknownvar}, which is proved in the following theorem.
\begin{thm}
\label{thm_sqrtn_esti_2}
    Under Assumptions \ref{assum_noise}-\ref{assum_cohypersphere} , it holds that 
    \begin{equation}
        \widehat{p}_n-p^0 = O_p(1/\sqrt{n}).
    \end{equation}
\end{thm}

\section{Algorithm for RSS-based localization}
\label{sec:algo}
 In this section, we summarize the algorithm for source localization using the proposed two-step estimator based on RSS measurements in Algorithm \ref{algo_rss}.


\begin{algorithm}
\caption{The algorithm for obtaining consistent and asymptotically efficient estimator using RSS measurements}  
\label{algo_rss}
\begin{algorithmic}[1]
\Require Sensor locations $\{p_i\}_{i=1}^n$, RSS measurements $\{10\logten(P_i)\}_{i=1}^n$, constants $P_0$, $\alpha$,   and noise variance $\sigma^2$ (if available).
\State Calculate the equivalent measurements $y_i = (\logten(P_i)-\logten(P_0))/(-\alpha)$ for $i=1,...,n$;
\If {$\sigma^2$ is available}
\State Calculate $b$ according to Lemma \ref{lem_eta};
\State Calculate the LS estimate $\widehat{\theta}_n$ according to \eqref{hattheta_knownvar};
\State Extract the estimate $\widehat{p}_n$ according to \eqref{hatp_knownvar}
\Else
\State Calculate the LS estimate $\widehat{\beta}_n$ according to \eqref{hatbeta_unknownvar};
\State Extract the estimate $\widehat{p}_n$ according to \eqref{hatp_unknownvar} with denominator replaced by $\max\{1,[\widehat{\beta}_n]_{m+2}\}$;
\EndIf

\State Run a one-step GN iteration \eqref{gn} for $\widehat{p}_n$ to obtain $\widehat{p}_n^{\rm GN}$;
\Ensure The source location estimate $\widehat{p}_n^{\rm GN}$.
\end{algorithmic}
\end{algorithm}

In Algorithm \ref{algo_rss}, to ensure numerical stability when handling small sample sizes, we replace the denominator \( [\widehat{\beta}_n]_{m+2} \) in \eqref{hatp_unknownvar} with \( \max\{1, [\widehat{\beta}_n]_{m+2}\} \), as suggested by Lemma \ref{lem_eta} where \( b > 1 \). This adjustment does not impact the asymptotic properties since \( [\widehat{\beta}_n]_{m+2} = b + O_p(1/\sqrt{n}) \) is larger than 1 with probability converging to 1.

For Algorithm \ref{algo_rss}, where \( X \in \mathbb{R}^{n \times (m+1)} \), \( \Phi \in \mathbb{R}^{n \times (m+2)} \) are the design matrices and \( Y \in \mathbb{R}^{n \times 1} \), \( \Gamma \in \mathbb{R}^{n \times 1} \) are the response vectors, each of the key steps (Lines 4, 7, and 10) has a computational complexity of \( \mathcal{O}(n) \). Since these steps dominate the algorithm's runtime and scale linearly with the number of measurements \( n \), the overall computational complexity of Algorithm \ref{algo_rss} is \( \mathcal{O}(n) \). This implies that the algorithm's execution time grows linearly with the data size.


\section{Simulations}
\label{sec:sim}

In this section, we perform Monte Carlo simulations to validate our theoretical results. Simulations are carried out separately for both 2-D and 3-D scenarios. In each scenario, we conduct \( N \) Monte Carlo trials with different realizations of measurement noise. For each estimator, the estimate obtained in the \( j \)-th trial is denoted by \( \widehat{p}_{n,j} \).

We evaluate the bias and root mean squared error (RMSE) of each estimator using the following definitions \cite{Wang2023, Zeng2024}:
\begin{align*}
    &\mathrm{Bias}(\widehat{p}_n) = \sum_{i=1}^m \big|\big[\Delta(\widehat{p}_n)\big]_i\big|, \quad \Delta(\widehat{p}_n) =   \frac{1}{N} \sum_{j=1}^N \widehat{p}_{n,j} - p^0  , \\
    &\mathrm{RMSE}(\widehat{p}_n) = \sqrt{\frac{1}{N} \sum_{j=1}^N \| \widehat{p}_{n,j} - p^0 \|^2},
\end{align*}
where \( p^0 \) denotes the true source location, and \( m \) is the dimension of the source coordinates, with \( m = 2 \) for the 2-D case and \( m = 3 \) for the 3-D case.

We use the root Cramér-Rao lower bound (RCRLB) as a performance benchmark, evaluating whether the RMSE of each estimator approaches the RCRLB.

We compare our proposed estimators for RSS-based localization, labeled as LS for the first step and LS+GN for the second step, with the following methods:

\begin{enumerate}[(i)] \item Eigen: An estimator that reformulates the problem of finding the global minimum of an approximated ML objective function as an eigenvalue problem \cite{Larsson2025}. \item Improved-LLS: An estimator that first forms a linear least squares solution, which is then refined using a linear system derived from the first-order Taylor expansion of the norm equation \cite{So2011}. \item Iteration: An estimator that iteratively solves a system of range equations, where the real ranges between sensors and the source are approximated using a Bayesian method \cite{Coluccia2014}. \item WLS-avg: An estimator that constructs a linear system by differencing the range equations with their average. The LS solution of this system is then used to approximate the covariance for a WLS solution as a refinement \cite{Salman2014}. \item LSRE-SDP: An estimator that optimizes the LSRE objective function with range constraints using SDP \cite{Wang2019}. \item ML-SDP: An estimator that optimizes an approximated ML objective function with a norm constraint using SDP, where the solver SeDuMi is recommended \cite{Vaghefi2013}. \end{enumerate}

Additionally, we evaluate the unknown-noise-variance counterparts of the proposed LS and LS+GN estimators, denoted as LS(unknown $\sigma$) and LS(unknown $\sigma$)+GN, respectively. 

All algorithms are implemented in MATLAB and executed on an AMD EPYC 7543 32-Core Processor.

\subsection{2-D scenario: fixed sensors}
\label{simu:rss_2d_fix}
We set 10 fixed sensors with the coordinates: 
\begingroup\allowdisplaybreaks
\begin{align*}
    & p_1 = [0,20]^T,~p_2=[0,50]^T,~p_3=[50,50]^T,~p_4=[50,0]^T,\\
    &p_5=[50,-50]^T,~ p_6 = [0,-50]^T,~p_7=[0,-20]^T,\\&p_8=[-50,-50]^T,~p_9=[-50,0]^T,~p_{10}=[-50,50]^T.
\end{align*}
\endgroup
The coordinates of source location are $p^0=[70,30]^T$. 

 Hence, Assumption \ref{assum_coordinates} holds. Each sensor measures \( T \) rounds of i.i.d. observations in total. As \( T \) increases, the number of RSS measurements becomes large enough to reveal the asymptotic behavior of the estimators. Note that this setup is equivalent to placing \( T \) sensors at each of the 10 specified placement sites, resulting in a total of \( 10T \) sensors making \( n = 10T \) i.i.d. observations. It is straightforward to observe that as \( T \) grows, the sample distribution of these \( 10T \) sensors converges to a distribution \( \mu \) such that \( \mu(p) = 1/10 \) for \( p \) at any of the 10 placement sites, and \( \mu(p) = 0 \) otherwise. Therefore, Assumptions \ref{assum_cohyperplane} and \ref{assum_cohypersphere} hold.

 We first present the biases and RMSEs for varying numbers of measurements (i.e., varying \( T \)) for each estimator. The standard deviation of the measurement noise is set to \( \sigma = 2 \, \text{dB} \). We evaluate \( T \) values of \( 3, 10, 30, 100, 200, \) and \( 400 \), conducting 1000 Monte Carlo runs for each \( T \). The biases for varying \( T \) are shown in Fig. \ref{fig:rss_2d_bias}, where it can be observed that the biases of the LS and LS+GN estimators are relatively small and converge to zero, confirming that both estimators are asymptotically unbiased. The WLS-avg estimator exhibits similar behavior to our proposed estimators mainly because its unconstrainted LS in the first step is unbiased, while the biases of the LSRE-SDP estimator also decrease but remain relatively larger. Other estimators exhibit significant biases that do not diminish with increasing sample sizes, indicating that they are biased. 
 
\begin{figure}[t]
    \centering
    \includegraphics[width=0.5\textwidth]{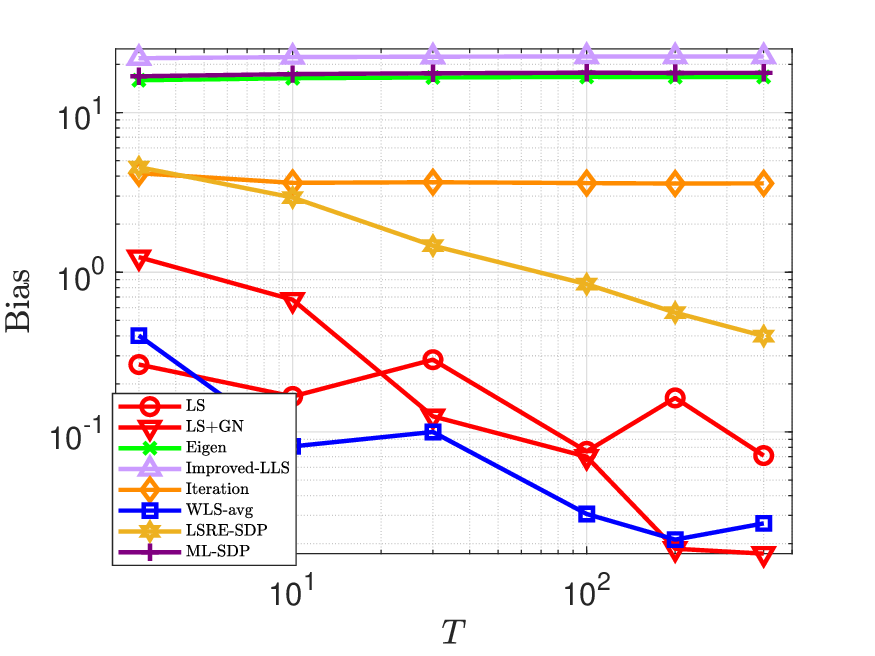}
    \caption{Biases of all the estimators with varying numbers of measurements for   2-D RSS-based localization. }
    \label{fig:rss_2d_bias}
\end{figure}

The RMSEs of these estimators for varying \( T \) are shown in Fig. \ref{fig:rss_2d_RMSE}. When the sample size is small, the LS+GN, Iteration, WLS-avg, and LSRE-SDP estimators exhibit similar RMSEs. As the sample size increases, the RMSE of the LS+GN estimator quickly converges to the RCRLB, confirming its asymptotic efficiency. The log-log plot in Fig. \ref{fig:rss_2d_RMSE} clearly demonstrates the  convergence rate \( O_p(1/\sqrt{n}) \) of both the RCRLB and the LS+GN estimator. While the RMSEs of the LS, Iteration, WLS-avg, and LSRE-SDP estimators also decrease, they remain higher than those of the LS+GN estimator. The remaining three estimators exhibit larger and non-decreasing RMSEs, primarily due to their biases.  

\begin{figure*}[!t]
	\centering
	\begin{subfigure}[b]{.48\textwidth}
    \centering
    \includegraphics[width=\textwidth]{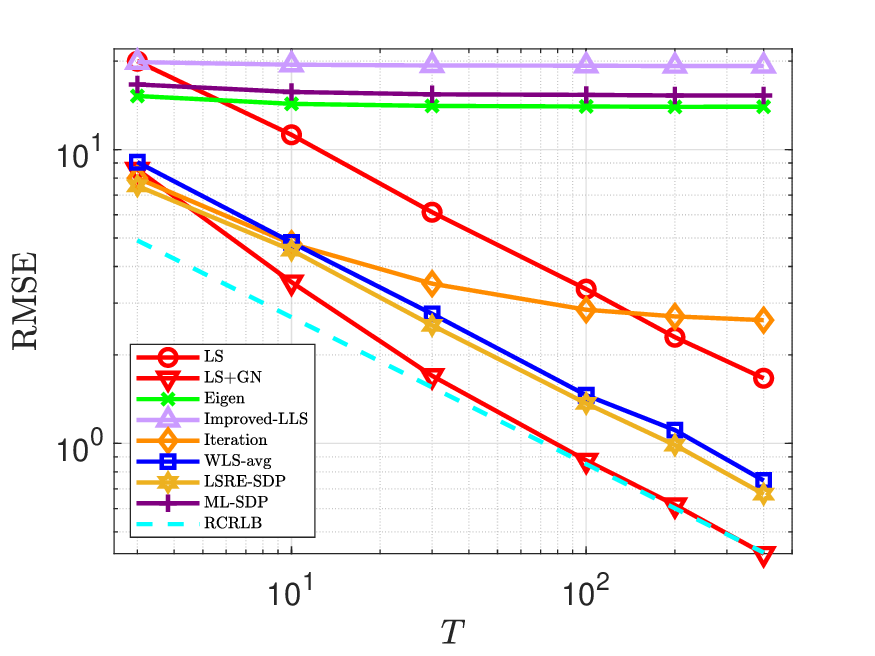}
    \caption{Varying numbers of measurements. }
    \label{fig:rss_2d_RMSE}
	\end{subfigure}
	\begin{subfigure}[b]{.48\textwidth}
    \centering
    \includegraphics[width=\textwidth]{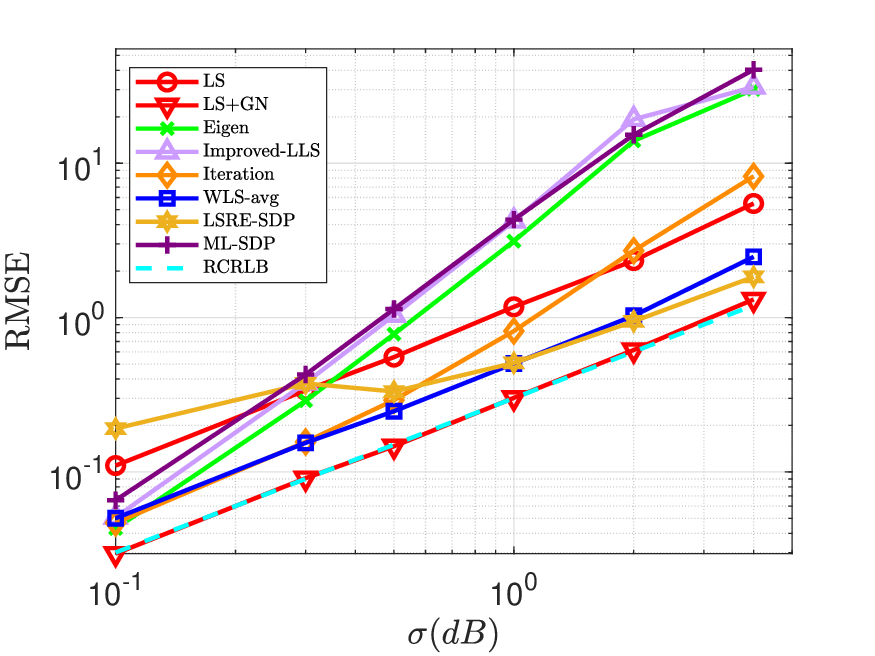}
    \caption{Varying noise intensities.}
    \label{fig:rss_2d_RMSE_varyingsigma}
	\end{subfigure}
	\caption{RMSEs for RSS-based localization  of 2-D fixed sensors.}
\end{figure*}

 Next, we investigate the RMSEs of each estimator for varying noise intensities in the case of a large sample size. We set \( \sigma \) to \( 0.1, 0.3, 0.5, 1, 2, \) and \( 4 \, \text{dB} \), respectively. For each \( \sigma \), we set \( T = 200 \) and conduct 1000 Monte Carlo runs. The RMSEs of the estimators for varying noise intensities are presented in Fig. \ref{fig:rss_2d_RMSE_varyingsigma}, where it can be observed that all estimators exhibit small RMSEs for low noise intensities. As the noise intensities increase, only the RMSEs of the LS+GN estimator approach the RCRLB, except for the case of extremely high noise intensity.

We further compare the computational time of these estimators as the number of measurements increases. Each algorithm is run 1000 times to obtain the average time cost, with the results summarized in Table \ref{table:rss_time}. For large sample sizes, our proposed algorithm achieves the fastest execution time. For smaller sample sizes, Improved-LLS has the lowest computational cost, though the difference between LS+GN and Improved-LLS is negligible. Both LS+GN and Improved-LLS involve solving a few small linear systems (e.g., \(2 \times 2\) or \(3 \times 3\) matrix inversions), resulting in \(\mathcal{O}(n)\) complexity. The Iteration method also scales as \(\mathcal{O}(n)\) but incurs higher overhead due to its iterative procedure. In contrast, WLS-avg requires inverting an \(n \times n\) approximated covariance matrix, leading to \(\mathcal{O}(n^3)\) complexity. The Eigen method involves \(n^2\)-element summations, matrix diagonalization, and eigenvalue decomposition, making it significantly slower for large \(n\). Finally, both LSRE-SDP and ML-SDP rely on the CVX toolbox to solve SDP problems, which is computationally expensive.
 
\begin{table}[ht]
\centering
\caption{The average time spent by different RSS-based localization algorithms among 1000 experiments. (Unit: seconds)}
\resizebox{1\columnwidth}{!}{
\begin{tabular}{ccccccc}
\hline
\hline
& $n=30$  & $n=100$  & $n=300$  & $n=1000$ & $n=2000$ & $n=4000$ \\ \hline
\textbf{LS+GN} & 0.00023& 0.00027& 0.00099& \textbf{0.00241}& \textbf{0.00443}& \textbf{0.00952} \\ \hline
Eigen & 0.00097& 0.00878& 0.07715& 0.85296& 3.44422& 14.38152 \\ \hline
Improved-LLS & \textbf{0.00012}& \textbf{0.00015}& \textbf{0.00094}& 0.00510& 0.02382& 0.02381 \\ \hline
Iteration & 0.00155& 0.00348 & 0.01772& 0.04462& 0.07316& 0.16394 \\ \hline
WLS-avg & 0.00018 & 0.00036 & 0.00331& 0.02263 & 0.08767 & 0.60154 \\ \hline
LSRE-SDP & 0.37977& 0.67624 & 1.56392& 4.90455& 10.37219& 23.88791 \\ \hline
ML-SDP & 0.26775& 0.31723 & 0.48078 & 1.01222 & 1.88521 & 3.86127 \\ \hline
\end{tabular}}
\label{table:rss_time}
\end{table}

Finally, we compare the LS and LS+GN estimators with their unknown-noise-variance counterparts, the LS(unknown $\sigma$) and LS(unknown $\sigma$)+GN estimators. The RMSEs for varying numbers of measurements are shown in Fig. \ref{fig:rss_compare_varyingT}, while those for varying noise intensities are displayed in Fig. \ref{fig:rss_compare_varyingsigma}. Although the unknown-noise-variance version of the LS estimator exhibits slightly higher RMSEs than the standard LS estimator, both the LS+GN estimator and its unknown-noise-variance variant asymptotically achieve the RCRLB with negligible differences—except under extremely high noise intensity. These results demonstrate the reliability of the proposed estimators even when the noise variance is unknown.
\begin{figure*}[!t]
	\begin{subfigure}[b]{.48\textwidth}
    \centering
    \includegraphics[width=\textwidth]{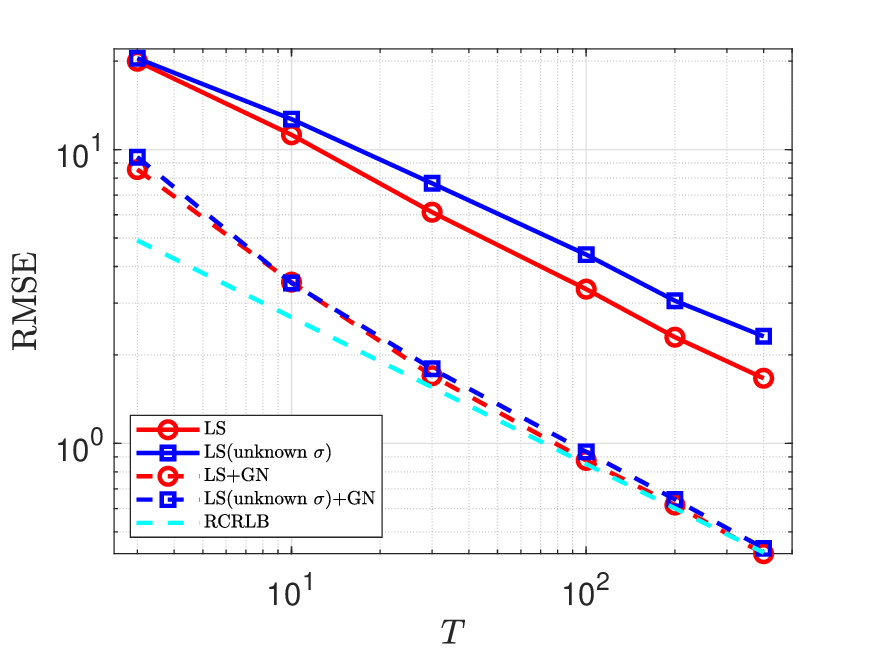}
    \caption{Varying numbers of measurements.}
    \label{fig:rss_compare_varyingT}
\end{subfigure}
\begin{subfigure}[b]{.48\textwidth}
   \centering
    \includegraphics[width=\textwidth]{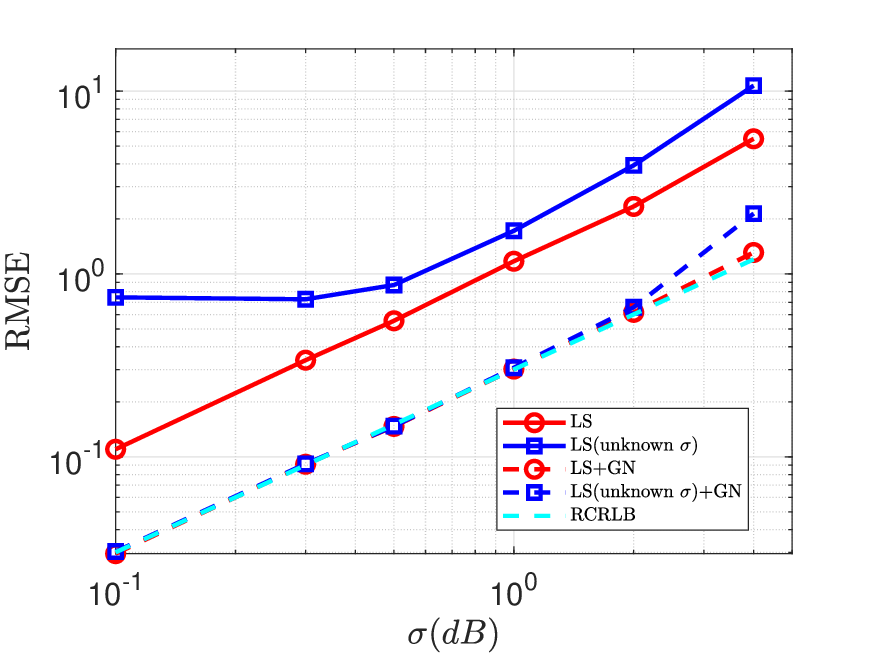}
    \caption{Varying noise intensities.}
    \label{fig:rss_compare_varyingsigma}
\end{subfigure}
\caption{RMSE comparison of the LS and LS+GN estimators with their unknown-noise-variance counterparts for RSS-based localization of 2-D fixed sensors.}
\label{fig:rss_compare}
\end{figure*}

\subsection{2-D scenario: random sensors}
 In this subsubsection, we consider random sensor deployment. The source's coordinates are \( p^0 = [120, 20]^T \). The number of sensors is set to \( n = 100, 300, 1000, 2000, 3000, \) and \( 4000 \), respectively. For each \( n \), we randomly generate \( n \) sensors \( \{ p_i = [x_i, y_i]^T, \, i = 1, \dots, n \} \), where \( \{ x_i, y_i, \, i = 1, \dots, n \} \) are independently sampled from a uniform distribution on \( [0, 100] \). This means the sensors are uniformly distributed within a square region with side length 100. It is straightforward to verify that Assumptions \ref{assum_coordinates} through \ref{assum_cohypersphere} hold. Note that the probability of having \( n \geq 100 \) sensors being collinear or concyclic is negligible. The noise is generated under Assumption \ref{assum_noise} with \( \sigma = 2 \, \text{dB} \), and each sensor makes one observation.

 We examine the biases and RMSEs for varying sample sizes of each estimator. For each \( n \), 1000 Monte Carlo experiments are conducted. The biases and RMSEs of each estimator are presented in Fig. \ref{fig:rss_bias_random} and Fig. \ref{fig:rss_rmse_random}, respectively. Similar to the 2-D fixed sensor scenario, the asymptotic unbiasedness and consistency of the LS and LS+GN estimators, as well as the asymptotic efficiency of the LS+GN estimator, are confirmed. Specifically, the LS, LS+GN, and WLS-avg estimators exhibit small biases that converge to zero. The Improved-LLS, Eigen, and Iteration estimators display significant and constant biases. However, the biases of the LSRE-SDP estimator increase as the sensors become more densely packed in the square region. As a result, the RMSEs of the LSRE-SDP estimator are close to the RCRLB for small sample sizes but become extremely high as the sensor density increases. In contrast, the RMSEs of the LS, LS+GN, and WLS-avg estimators converge with a rate of \( O_p(1/\sqrt{n}) \), with the RMSEs of the LS+GN estimator achieving the RCRLB.

\begin{figure*}[!t]
	\begin{subfigure}[b]{.48\textwidth}
    \centering
    \includegraphics[width=\textwidth]{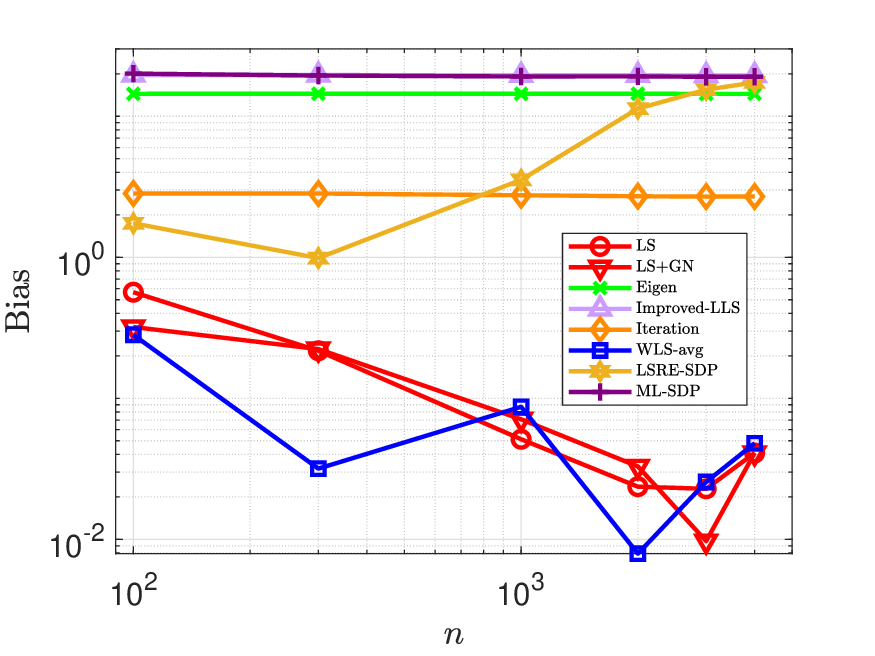}
    \caption{Biases for varying numbers of measurements.}
    \label{fig:rss_bias_random}
\end{subfigure}
\begin{subfigure}[b]{.48\textwidth}
   \centering
    \includegraphics[width=\textwidth]{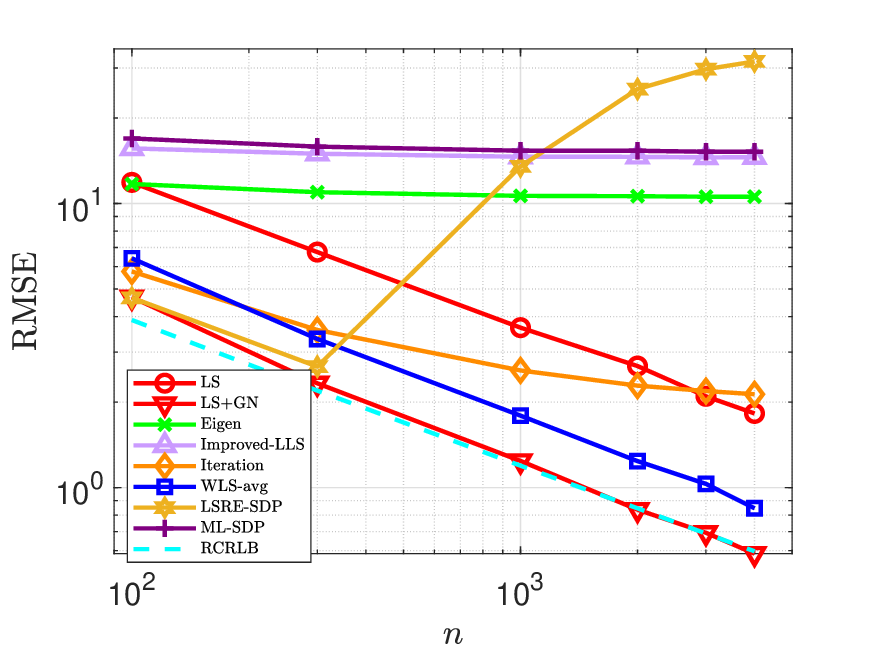}
    \caption{RMSEs for varying numbers of measurements.}
    \label{fig:rss_rmse_random}
\end{subfigure}
\caption{Biases and RMSEs for RSS-based localization   of 2-D random deployed sensors.}
\label{fig:2d_random}
\end{figure*}

\subsection{3-D scenario: fixed sensors}

 The settings for the 3-D scenario with fixed sensors are similar to those of the 2-D case, with the coordinates of the 10 fixed sensors given by:
\begingroup
\allowdisplaybreaks
\begin{align*}
    & p_1 = [0,20,50]^T,~p_2=[0,50,0]^T,~p_3=[50,50,-50]^T,\\&p_4=[50,0,0]^T,~p_5=[50,-50,50]^T,
    ~ p_6 = [0,-50,0]^T,\\&p_7=[0,-20,-50]^T,~p_8=[-50,-50,0]^T,\\&p_9=[-50,0,50]^T,~p_{10}=[-50,50,-50]^T,
\end{align*}
\endgroup
The coordinates of the source location are \( p^0 = [70, 30, 10]^T \). 
 
 In the Monte Carlo experiments, the settings for \( T \) and \( \sigma \) are consistent with those used in the 2-D case. The biases of each estimator for varying numbers of measurements are presented in Fig. \ref{fig:rss_3d_bias}, while the RMSEs for varying numbers of measurements and varying noise intensities are shown in Fig. \ref{fig:rss_3d_RMSE} and Fig. \ref{fig:rss_3d_RMSE_varyingsigma}, respectively. To compare the LS and LS+GN estimators with their unknown-noise-variance counterparts, Fig. \ref{fig:rss_3d_compare} illustrates the RMSEs of these estimators. Similar results are observed in these figures, and comparable conclusions can be drawn to those discussed in Subsection \ref{simu:rss_2d_fix}. 

\begin{figure}[t]
    \centering
    \includegraphics[width=0.5\textwidth]{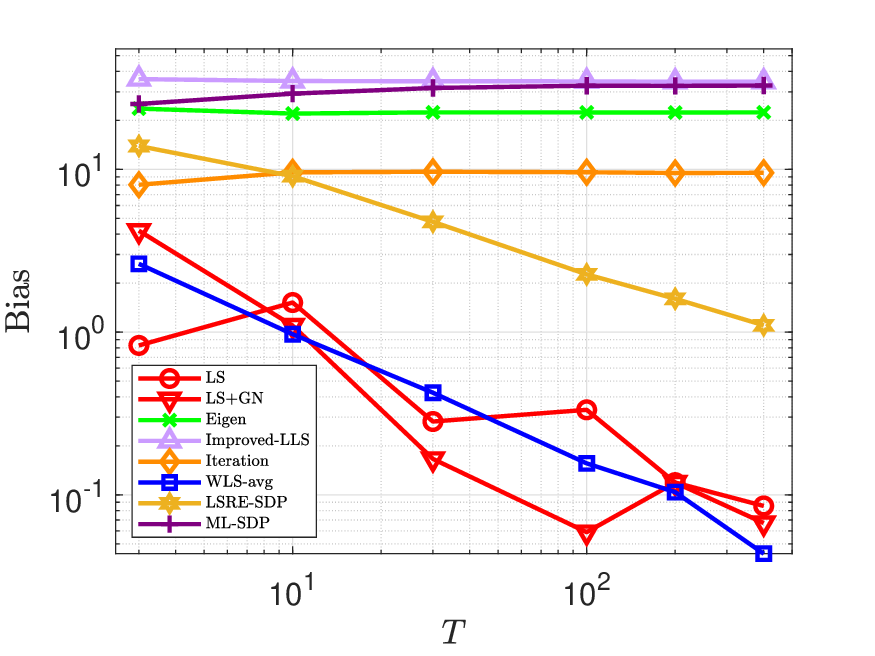}
    \caption{Biases of the estimators with varying numbers of measurements for RSS-based localization   of 3-D fixed sensors. }
    \label{fig:rss_3d_bias}
\end{figure}

\begin{figure*}[!t]
\label{fig6}
	\centering
	\begin{subfigure}[b]{.48\textwidth}
    \centering
    \includegraphics[width=\textwidth]{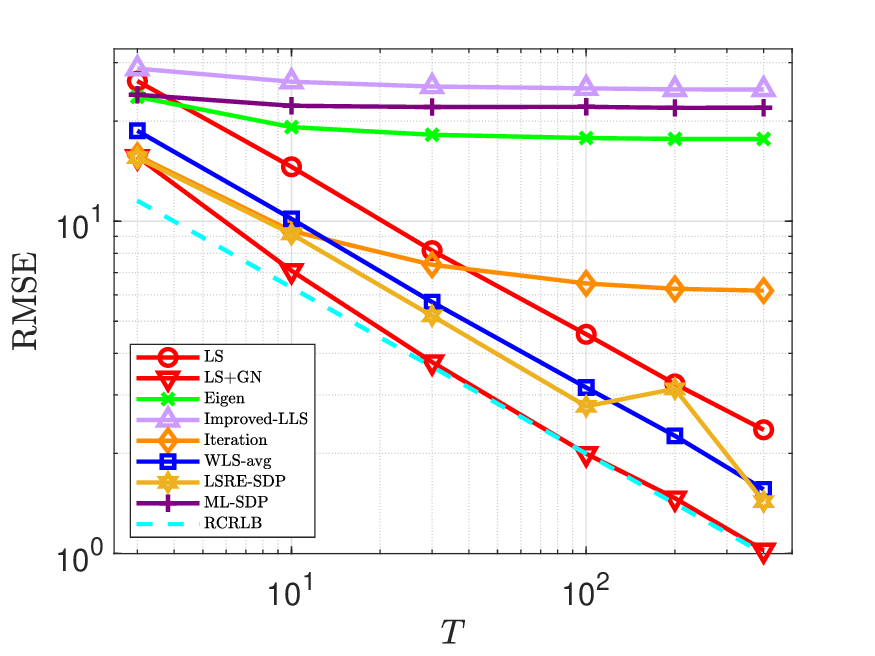}
    \caption{Varying numbers of measurements. }
    \label{fig:rss_3d_RMSE}
	\end{subfigure}
	\begin{subfigure}[b]{.48\textwidth}
    \centering
    \includegraphics[width=\textwidth]{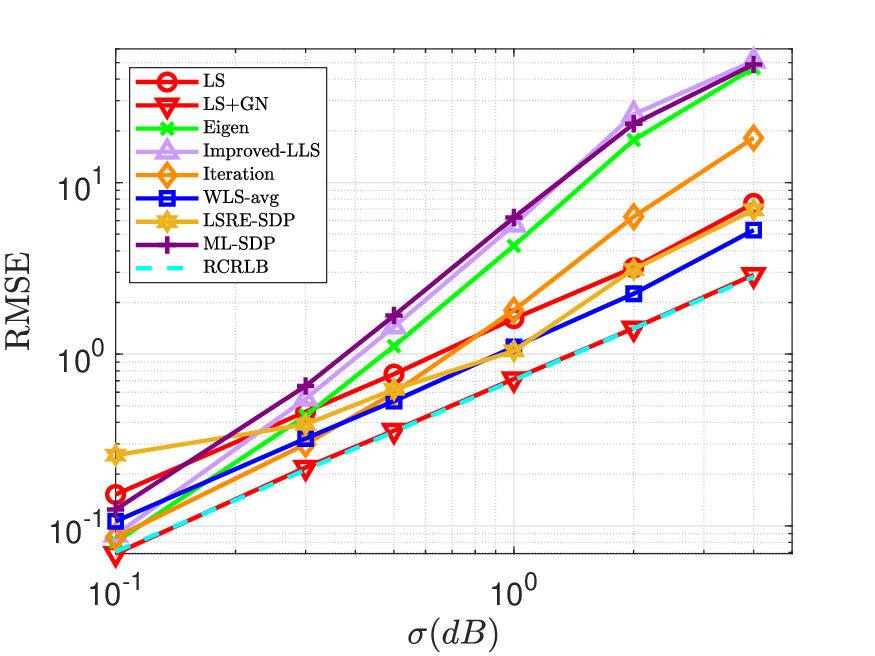}
    \caption{Varying noise intensities.}
    \label{fig:rss_3d_RMSE_varyingsigma}
	\end{subfigure}
	\caption{RMSEs for RSS-based localization  of 3-D fixed sensors.}
\end{figure*}

\begin{figure*}[!t]
	\begin{subfigure}[b]{.48\textwidth}
    \centering
    \includegraphics[width=\textwidth]{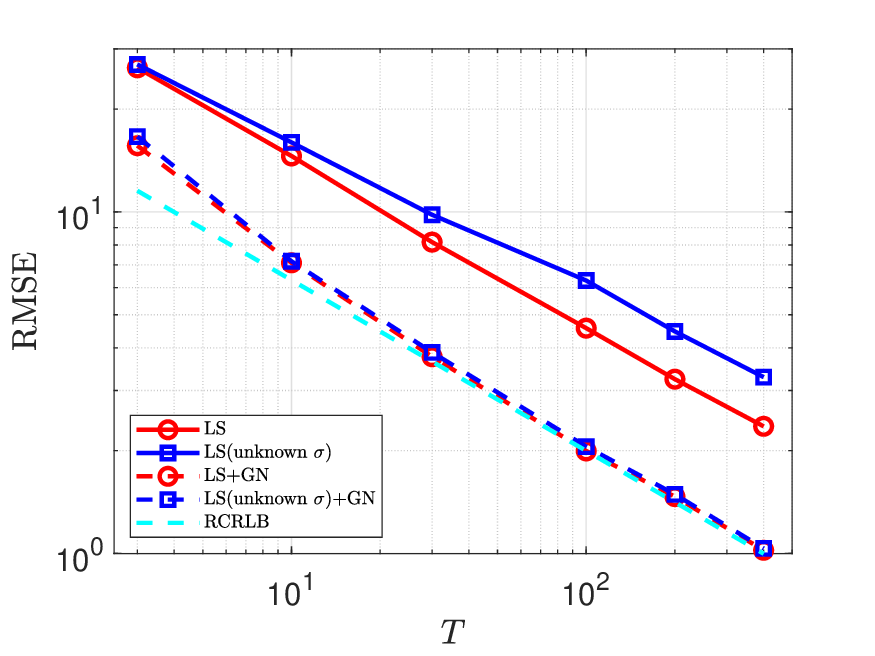}
    \caption{Varying numbers of measurements.}
    \label{fig:rss_3d_compare_varyingT}
\end{subfigure}
\begin{subfigure}[b]{.48\textwidth}
   \centering
    \includegraphics[width=\textwidth]{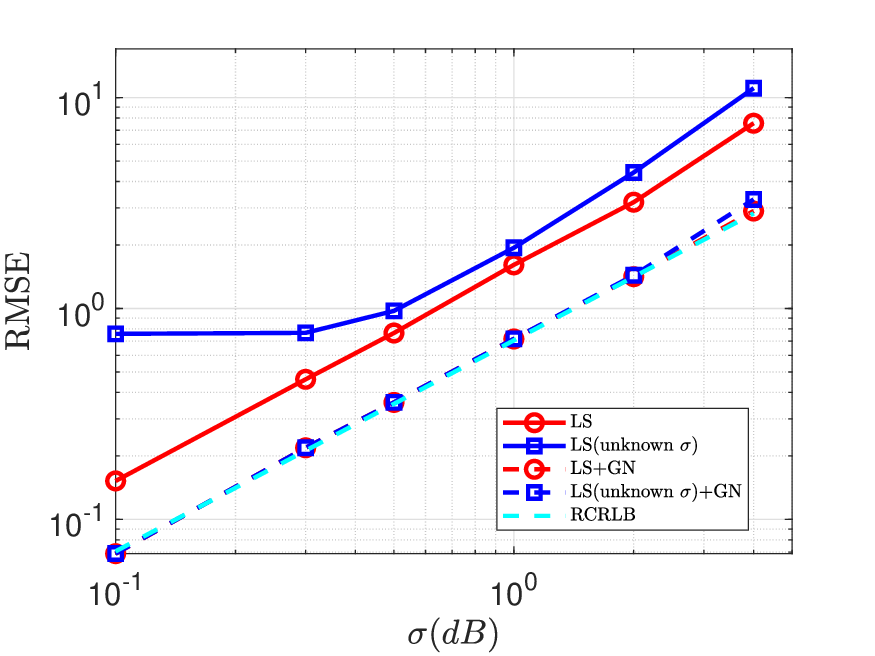}
    \caption{Varying noise intensities.}
    \label{fig:rss_3_dcompare_varyingsigma}
\end{subfigure}
\caption{RMSE comparison of the LS and LS+GN estimators with their unknown-noise-variance counterparts for RSS-based localization of 3-D fixed sensors.}
\label{fig:rss_3d_compare}
\end{figure*}

\section{Conclusion}
\label{sec:con}
 In this paper, we have established the asymptotic localizability of the RSS-based localization problem with respect to sensor deployment. We have also proposed a   consistent and asymptotically efficient two-step estimator for source localization using RSS measurements, under specific conditions related to measurement noise and sensor geometry.
In the first step, we derive \( \sqrt{n} \)-consistent estimators for the true source coordinates using the LS estimators, considering both known and unknown noise variance cases.
In the second step, we refine  the \( \sqrt{n} \)-consistent estimator from the first step by the GN iterative algorithm. Theoretically, the two-step estimator achieves asymptotic efficiency with just one-step GN iteration under appropriate conditions. Additionally, considering the computational complexity of the LS estimators and the one-step GN iteration, the total computational cost of the two-step estimator is 
$\mathcal{O}(n)$, making it computationally cheap for practical applications.
Monte Carlo simulations demonstrate that our two-step estimator achieves both statistical efficiency and low computational complexity, outperforming existing methods for RSS-based localization.

 \section*{Appendix A: Proofs of results}

This section contains the proofs of the results in the main body of the paper.
\renewcommand{\thesection}{A}
\setcounter{equation}{0}
\renewcommand{\theequation}{A\arabic{equation}}
\setcounter{lem}{0}
\setcounter{subsection}{0}

\subsection{Proof of Theorem \ref{lem_uniq}}
\label{pf_uniq}
Under Assumptions \ref{assum_coordinates} and \ref{assum_cohyperplane} (i), we have $\mathcal{P} \times \mathcal{P}^0$ is compact and $(\logten(\|\tilde{p}-p\|) - \logten(\|\tilde{p}-p^0\|))^2$ is a continuous and bounded function of $(\tilde{p},p) \in \mathcal{P} \times \mathcal{P}^0$. Then by Lemma \ref{lem_converge_samplemean}, we have
$
    \lim_{n\to \infty} h_n(p) = \mathbb{E}_\mu\left[(\logten(\|\tilde{p}-p\|) - \logten(\|\tilde{p}-p^0\|))^2\right],
$
where $\mathbb{E}_\mu$ is taken over $\tilde{p}\in\mathbb{R}^m$ with respect to the distribution $\mu$. Suppose that there exists some $p \neq p^0$ such that
$
    \mathbb{E}_\mu\left[(\logten(\|\tilde{p}-p\|) - \logten(\|\tilde{p}-p^0\|))^2\right] = 0.
$
For every such $p$, define $\mathcal{P}_p \eq \{ \tilde{p} \in \mathcal{P} \vert  \logten(\|\tilde{p}-p\|) = \logten(\|\tilde{p}-p^0\|)\}$. Then $\mu(\mathcal{P}_p) =1$. However, $\logten(\|\tilde{p}-p\|) = \logten(\|\tilde{p}-p^0\|)$ is equivalent to $\|\tilde{p}-p\| = \|\tilde{p}-p^0\|$, which means
$
    (p^0-p)^T\left(\tilde{p} - \frac{1}{2}(p^0+p)\right) = 0.
$
Since \(p \neq p^0\), it follows that every \(\tilde{p} \in \mathcal{P}_p\) lies on the hyperplane which is perpendicular to the vector \(p^0 - p\) and passes through the midpoint \((p^0 + p)/2\). This contradicts Assumption \ref{assum_cohyperplane}.

This completes the proof.

\subsection{Proof of Lemma \ref{lem_converge_M}}
\label{pf_converge_M}
    It is straightforward that
    \begin{equation*}
        \frac{1}{n}\sum_{i=1}^n (\nabla f_i(p))\nabla f_i(p)^T = \frac{1}{n}\sum_{i=1}^n\frac{(p-p_i)(p-p_i)^T}{\|p_i-p\|^4\ln^2(10)}.
    \end{equation*}
Under Assumptions \ref{assum_coordinates}-\ref{assum_cohyperplane}, we have 
$
    \frac{1}{\|\tilde{p}-p\|^4 \ln^2(10)} (p-\tilde{p})(p-\tilde{p})^T
$
is a continuous and bounded matrix function of $(\tilde{p},p) \in \mathcal{P}\times \mathcal{P}^0$. Then by Lemma \ref{lem_converge_samplemean}, we  obtain that $\frac{1}{n}\sum_{i=1}^n\big(\nabla f_i(p)\big) \nabla f_i(p)^T$ converges uniformly on $\mathcal{P}^0$  as $n \to \infty$. This proves point (i).
 
 Moreover, by the definition of $M^0$, we have
    \begin{align*}
        M^0  
        =\mathbb{E}_\mu \left[\frac{1}{\|\tilde{p}-p^0\|^4\ln^2(10)}(p^0-\tilde{p})(p^0-\tilde{p})^T\right],
    \end{align*}
where $\mathbb{E}_\mu$ is taken over $\tilde{p}$ with respect to $\mu$.
Let $\psi$ be a vector such that
\begin{align*}
    \psi^T M^0\psi =& \mathbb{E}_\mu \left[ \frac{1}{\|\tilde{p}-p^0\|^4\ln^2(10)}\psi^T(p^0-\tilde{p})(p^0-\tilde{p})^T\psi\right]\\
        =& \mathbb{E}_\mu \left[\frac{1}{\|\tilde{p}-p^0\|^4\ln^2(10)}\left((p^0-\tilde{p})^T\psi\right)^2\right]= 0.
\end{align*}
For every such $\psi$, define $\mathcal{P}_\psi \eq \{\tilde{p} \in \mathcal{P} \vert (p^0-\tilde{p})^T\psi = 0\}$. Then $\mu(\mathcal{P}_\psi) =1$ and derives that $\psi=0$. Otherwise,  $\psi \neq 0$ implies that   every $\tilde{p} \in \mathcal{P}_\psi$  lies on a hyperplane perpendicular to $\psi$, which contradicts Assumption \ref{assum_cohyperplane}.
This proves point (ii).

This completes the proof.

\subsection{Proof of Proposition \ref{thm_converge_ML_objfunc}}
\label{pf_converge_ML_objfunc}
By the definition of $\ell_n(p)$, we have
\begingroup
\allowdisplaybreaks
\begin{align*}
    &\ell_n(p)/n\\
    &= -\ln\left(\frac{\sqrt{2\pi}\sigma}{10\alpha}\right)-\frac{50\alpha^2}{\sigma^2}\frac{1}{n}\sum_{i=1}^n(f_i(p^0)-f_i(p)+\omega_i)^2\\
    &=-\ln\left(\frac{\sqrt{2\pi}\sigma}{10\alpha}\right)-\frac{50\alpha^2}{\sigma^2}\frac{1}{n}\sum_{i=1}^n\big(2\omega_i(f_i(p^0)-f_i(p))+\omega_i^2\big)\\
    &~~~~-\frac{50\alpha^2}{\sigma^2}h_n(p).
\end{align*}
\endgroup
Since $\mathcal{P}$ and $\mathcal{P}^0$ are bounded under Assumption \ref{assum_coordinates}, by Lemma \ref{lem_sqrtn}, it holds that $\frac{1}{n}\sum_{i=1}^n(f_i(p^0)-f_i(p))\omega_i \xra{} 0$ almost surely uniformly on $\mathcal{P}^0$. And it follows from Lemma \ref{lem_sqrtn}  that $\lim_{n \to \infty}\frac{1}{n}\sum_{i=1}^n \omega_i^2 = \sigma^2/(100\alpha^2)$  almost surely. Therefore, by Lemma \ref{lem_uniq}, we have
\begin{align*}
     \frac{1}{n}\ell_n(p) \xra{} -\ln(\sqrt{2\pi}\sigma/(10\alpha)) -\frac{50\alpha^2}{\sigma_a^2}h(p)-\frac{1}{2} = \ell(p)
\end{align*}
almost surely as $n\ra{}\infty$ uniformly for $p \in \mathcal{P}^0$. 

Next, we show that $\nabla^2(-\ell(p^0)) = 100\alpha^2M^0/\sigma^2$. Note that $\nabla(-\ell_n(\widehat{p}_n^{\rm ML})) = 0$. The Taylor expansion of $-\ell_n(p)/n$ in a small neighborhood of $\widehat{p}_n^{\rm ML}$ is \begin{align*}
    -\ell_n(p)/n =&-\ell_n(\widehat{p}_n^{\rm ML})/n \\&+ \frac{1}{2}(p - \widehat{p}_n^{\rm ML})^T \left( \nabla^2(-\ell_n(\widehat{p}_n^{\rm ML})/n) \right)(p - \widehat{p}_n^{\rm ML})\\&+ o_p(\|p - \widehat{p}_n^{\rm ML}\|^2),
\end{align*}
where \begin{align*}
    &\nabla^2(-\ell_n(p)/n) \\=&~ \frac{100\alpha^2}{\sigma^2 n}\sum_{i=1}^n \left[\big(\nabla f_i(p)\big)^T \nabla f_i(p) - \nabla^2 f_i(p)(y_i-f_i(p))\right].
\end{align*} Note that
$
   \nabla^2(-\ell_n(p^0)/n)  \to 100\alpha^2M^0/\sigma^2
$ almost surely.
Morevoer, we can prove that $\nabla^2(-\ell_n(p)/n)$ converges uniformly on $\mathcal{P}^0$ using Lemma \ref{lem_converge_samplemean} as for proving Assumptions \ref{assum_coordinates}-\ref{assum_cohyperplane} and Lemma \ref{lem_converge_M}.
Then, as a direct collary of the uniform convergence, we have  $\nabla^2(-\ell_n(p^0)/n) \to \nabla^2 (-\ell(p^0))$ as $n \to \infty$, which implies $\nabla^2(-\ell(p^0)) = 100\alpha^2M^0/\sigma^2$. 

This completes the proof.

\subsection{Proof of Lemma \ref{lem_eta}}
\label{pf_eta}
 It is straightforward that $2\omega_i = \logten(10^{2\omega_i}) = \ln(10^{2\omega_i})/\ln10,~i=1,...,n$. Since $\{\omega_i\}_{i=1}^n$ are i.i.d. Gaussian random variables, by Lemma \ref{lem_lognormal} , one derives that $\{10^{2\omega_i}\}_{i=1}^n$ are i.i.d. lognormal random variables. Note that $\mathbb{E}[2(\ln10)\omega_i] = 0$ and $\mathbb{V}[2(\ln10)\omega_i] = (\ln 10)^2\sigma^2/(25\alpha^2),~i=1,...,n$. By Lemma \ref{lem_lognormal}, we have
\begin{align*}
     b = \mathbb{E}(10^{2\omega_i}) =e^{(\ln 10)^2\sigma^2/(50\alpha^2)}, 
     \mathbb{V}(10^{2\omega_i})  = b^2(b^2-1).
\end{align*}

    This completes the proof.
    
\subsection{Proof of Proposition \ref{prop_invertgram_1}}
\label{pf_invertgram_1}
    Let $\psi = [\psi_{1:m}^T,~\psi_{m+1}]^T$ be a vector such that
$
    \psi^TX^TX\psi/n = 0.
$
Then it holds that  $-2p_i^T\psi_{1:m}+\psi_{m+1} = 0$ for all $i=1,...,n$. 
This derives that $\psi = 0$. Otherwise,  $\psi \neq  0$  leads to a contradiction that all the sensors $\{p_i,~i=1,...,n\}$ lie on an identical hyperplane perpendicular to $\psi_{1:m}$, i.e.,   $(p_i-p_1)^T\psi_{1:m} = 0$  for all $i=1,...,n$. This  contradicts Assumption \ref{assum_cohyperplane}. Thus, $X^TX/n$ is nonsingular.

Moreover, it follows that  
\begin{align*}
    \frac{1}{n}X^TX  
    =\frac{b^2}{n} \sum_{i=1}^n\begin{bmatrix}4p_ip_i^T & -2p_i\\-2p_i^T & 1\end{bmatrix}.
\end{align*}
It is obvious that the matrix
$\begin{bmatrix}4\tilde{p}\tilde{p}^T & -2\tilde{p}\\-2\tilde{p}^T & 1\end{bmatrix}$ is a continuous and bounded matrix function with respect to $\tilde{p} \in \mathcal{P}$. Then by Lemma \ref{lem_converge_samplemean}, it holds that $\lim_{n \to \infty}X^TX/n$ exists and 
$$\lim_{n \to \infty}X^TX/n = \mathbb{E}_\mu 
\begin{bmatrix}4\tilde{p}\tilde{p}^T & -2\tilde{p}\\-2\tilde{p}^T & 1\end{bmatrix},$$
where $\mathbb{E}_\mu$ is taken over $\tilde{p} \in \mathcal{P}$ with respect to   $\mu$.

Lastly, we prove the non-singularity of $\lim_{n \to \infty}X^TX/n$. Let $\psi' = [\psi_{1:m}'^T,~\psi_{m+1}']^T  $ be a vector such that 
$(\psi')^T\Big(\lim_{n\to \infty}\frac{1}{n}\Phi^T\Phi\Big)\psi'=0.$
Then it holds that
\begin{equation*}
    b^2\mathbb{E}_\mu \left[ 4(\tilde{p}^T\psi_{1:m}')^2 - 4\psi_{m+1}'\tilde{p}^T\psi_{1:m}' + (\psi_{m+1}')^2 \right]=0.
\end{equation*}
Define $\mathcal{P}_{\psi'} \eq  \left\{\tilde{p} \in \mathcal{P} \vert 2\tilde{p}^T\psi_{1:m}' - \psi_{m+1}' = 0\right\}$ for every such $\psi'$. Under Assumption \ref{assum_noise}, by Lemma \ref{lem_eta}, we have $b>0$. Then it holds that $\mu(\mathcal{P}_{\psi'}) = 1$ and further $\psi'=0$.
Otherwise, when $\psi'\neq 0$, the equation
$2\tilde{p}^T\psi_{1:m}' - \psi_{m+1} = 0$  means all the $\tilde{p} \in \mathcal{P}_{\psi'}$ lie on an identical hyperplane. This contradicts Assumption \ref{assum_cohyperplane} and proves the non-singularity of $\lim_{n \to \infty}X^TX/n$.

This completes the proof

\subsection{Proof of Theorem \ref{thm_sqrtn_esti_1}}
\label{pf_sqrtn_esti_1}
We first show that $\widehat{\theta}_n$ is a $\sqrt{n}$-consistent estimator of $\theta^0$. It is straightforward that
\begin{equation*}
    \widehat{\theta}_n = (X^TX)^{-1}X^T(X^T\theta^0+V) = \theta^0+\frac{1}{n}\left(\frac{1}{n}X^TX\right)^{-1}X^TV.
\end{equation*}
Note that 
$
    \frac{1}{n}X^TV = \frac{b}{n}\begin{bmatrix} -2\sum_{i=1}^np_id_i^2\eta_i \\ \sum_{i=1}^nd_i^2\eta_i  \end{bmatrix}.
$
Under Assumption \ref{assum_coordinates}, all the sensors  $p_i$  and distances  $d_i^2$ are uniformly bounded for all $i=1,...,n$. And by Lemma \ref{lem_eta}, it holds that $\{\mathbb{E}(\eta_i^2)\}_{i=1}^n$ are uniformly upper bounded. Then we have, for example,
$
    \mathbb{V}(d_i^2\eta_i) \leq \max_{i=1,...,n}\{d_i^2\} \mathbb{V}(\eta_i).
$
Then by Lemma \ref{lem_sqrtn}, we have
\begin{equation*}
    \frac{1}{n}\sum_{i=1}^np_id_i^2\eta_i = O_p(1/\sqrt{n}),
\quad 
    \frac{1}{n}\sum_{i=1}^nd_i^2\eta_i = O_p(1/\sqrt{n}).
\end{equation*}
This derives that
$
    \frac{1}{n}X^TV= O_p(1/\sqrt{n}).
$
and further $
    \widehat{\theta}_n-\theta^0 = O_p(1/\sqrt{n})
$  by Proposition \ref{prop_invertgram_1}.
We reach
$
    \widehat{p}_n-p^0=O_p(1/\sqrt{n}).
$

This completes the proof.

\subsection{Proof of Proposition \ref{prop_invertgram_2}}
\label{pf_invertgram_2}
With a little of abuse of notation, 
let $\psi = [\psi_{1:m}^T,~\psi_{m+1},~\psi_{m+2}]^T$ be a vector such that
$
    \frac{1}{n}\psi^T\Phi^T\Phi\psi = 0.
$
Then it holds that
\begin{equation*}
    -2p_i^T\psi_{1:m} + \psi_{m+1} + \|p_i\|^2 \psi_{m+2} = 0,\quad i=1,...,n.
\end{equation*}
This derives that  $\psi_{m+2} = 0$.
Otherwise, i.e., $\psi_{m+2} \neq 0$, it holds that
\begin{equation*}
 \left\|p_i-\frac{\psi_{1:m}}{\psi_{m+2}}\right\|^2= \frac{\|\psi_{1:m}\|^2-\psi_{m+1}\psi_{m+2}}{\psi_{m+2}^2},~ i=1,...,n.
\end{equation*}
This means that all the $n$ sensors lie on a circle when \(m = 2\) or on a sphere when \(m = 3\), which contradict Assumption \ref{assum_cohypersphere}.
Then, it follows that 
$
    -2p_i^T\psi_{1:m} + \psi_{m+1} = 0,~ i=1,...,n
$.
By the similar augments used in the proof of Proposition \ref{prop_invertgram_1}, we can derive  $\psi_{1:m}=0,\psi_{m+1}= 0$. Thus, we obtain $\psi =0$.
This proves that $\Phi^T\Phi/n$ is nonsingular.

Moreover, by the similar augments used in the proof of Proposition \ref{prop_invertgram_1}, we can check that the limit $\lim_{n \to \infty}\Phi^T\Phi/n$ exists and it holds that
$$\lim_{n \to \infty}\frac{1}{n}\Phi^T\Phi = \mathbb{E}_\mu\begin{bmatrix} 4\tilde{p}\tilde{p}^T & -2\tilde{p} & -2\|\tilde{p}\|^2\tilde{p}\\ -2\tilde{p}^T & 1 & \|\tilde{p}\|^2\\ -2\|\tilde{p}\|^2\tilde{p}^T & \|\tilde{p}\|^2 & \|\tilde{p}\|^4 \end{bmatrix} ,$$
where $\mathbb{E}_\mu$ is taken over $\tilde{p} \in \mathcal{P}$ with respect to $\mu$.

Lastly, we prove the non-singularity of $\lim_{n \to \infty}\Phi^T\Phi/n$. 
Let $\psi' = [\psi_{1:m}'^T,~\psi'_{m+1},~\psi'_{m+2}]^T  $ be a vector such that
$(\psi')^T \Big(\lim_{n \to \infty}\frac{1}{n}\Phi^T\Phi \Big)\psi' = 0.$
Then it holds that 
$\mathbb{E}_\mu \left[ (-2\tilde{p}^T\psi_{1:m}' + \psi_{m+1}' + \|\tilde{p}\|^2\psi_{m+2}')^2 \right] = 0.$
Denote $\mathcal{P}_{\psi'} = \{\tilde{p} \in \mathcal{P} \vert -2\tilde{p}^T\psi_{1:m}' + \psi_{m+1}' + \|\tilde{p}\|^2\psi_{m+2}' = 0\}$ for every such $\psi'$. Then $\mu(\mathcal{P}_{\psi'}) = 1$. However, every $\tilde{p} \in \mathcal{P}_{\psi'}$ satisfies
$-2\tilde{p}^T\psi_{1:m}' + \psi_{m+1}' + \|\tilde{p}\|^2\psi_{m+2}' = 0.$
This derives that $\psi_{m+2}' = 0$.
Otherwise, $\psi_{m+2}' \neq 0$. Then every $\tilde{p} \in \mathcal{P}_{\psi'}$ satisfies
\begin{equation*}
 \left\|\tilde{p}-\frac{\psi'_{1:m}}{\psi'_{m+2}}\right\|^2= \frac{\|\psi'_{1:m}\|^2-\psi'_{m+1}\psi'_{m+2}}{(\psi'_{m+2})^2},
\end{equation*}
which contradicts Assumption \ref{assum_cohypersphere}.
Accordingly, we proved   $-2\tilde{p}^T\psi_{1:m}' + \psi_{m+1}' = 0.$
By Assumption \ref{assum_cohyperplane}, we can prove that 
 $\psi'_{1:m}=0,~\psi'_{m+1}=0$ as that used in the proof of Proposition \ref{prop_invertgram_1}.
 Thus, we reach $\psi'=0$ and
 this entails that $\lim_{n \to \infty}\Phi^T\Phi/n$ is nonsingular.

This completes the proof.

\subsection{Proof of Theorem \ref{thm_sqrtn_esti_2}}
\label{pf_sqrtn_esti_2}
Under Assumptions \ref{assum_noise}-\ref{assum_cohypersphere}, by the similar arguments used in the proof of Theorem \ref{thm_sqrtn_esti_1}, we have
$
    \widehat{\beta}_n - \beta^0 = O_{p}(1/\sqrt{n}).
$
Then it holds that $[\widehat{\beta}_n]_{1:m} = bp^0+O_p(1/\sqrt{n})$ and $[\widehat{\beta}_n]_{m+2} = b+O_p(1/\sqrt{n})>1$. 
Accordingly, we have
\begin{align*}
    \widehat{p}_n& = \frac{[\widehat{\beta}_n]_{1:m}}{[\widehat{\beta}_n]_{m+2}} = \frac{bp^0+O_p(1/\sqrt{n})}{b+O_p(1/\sqrt{n})} = \frac{p^0+O_p(1/\sqrt{n})}{1+O_p(1/\sqrt{n})}\\
    & = p^0 +O_p(1/\sqrt{n}).
\end{align*}
This completes the proof.

\section*{Appendix B: Auxiliary lemmas}
\renewcommand{\thesection}{B}
\setcounter{equation}{0}
\renewcommand{\theequation}{B\arabic{equation}}
\setcounter{lem}{0}
\setcounter{subsection}{0}
\renewcommand{\thelem}{B\arabic{lem}}

This subsection contains the auxiliary lemmas used for the proofs.
\begin{lem}\cite[Chapter 14.3]{Johnson1994}
\label{lem_lognormal}
Let $X$ be a Gaussian random variable with mean zero and variance $\lambda^2>0$. Then $e^X$ is a lognormal random variable with mean $
   e^{\frac{1}{2}\lambda^2} 
$ and variance 
$
  e^{\lambda^2}(e^{\lambda^2}-1)
$.
Moreover, $\mathbb{V}\big((e^{X}-\mathbb{E}(e^X))^2\big) < \infty$. 
\end{lem}

\begin{lem}
\label{lem_converge_samplemean}
    Let $\mu_n$ be an empirical distribution with a compact support $\mathcal{Q} \subset \mathbb{R}^m$, which converges to a distribution $\mu$. Then for any continuous and bounded function $f(\cdot)$ on $\mathcal{Q}$, it holds that $\mathbb{E}_\mu(f(x))$ exists and 
    \begin{equation*}
        \int f(x) d\mu_n(x) \to \int f(x) d\mu(x) = \mathbb{E}_\mu(f(x)),~n \to \infty,
    \end{equation*}
    where the expectation is taken over $x$ with respect to $\mu$.
    
    Further, for any continuous and bounded function $f(x,c)$ on $\mathcal{Q} \times \mathcal{Q}^0$, where $\mathcal{Q}^0$ is compact as well, it holds that $\mathbb{E}_\mu(f(x,c))$ exists and 
    \begin{equation*}
        \int f(x,c) d\mu_n(x) \to \int f(x,c) d\mu(x) = \mathbb{E}_\mu(f(x,c)),~ n \to \infty
    \end{equation*}
    uniformly on $\mathcal{Q}^0$,  where the expectation is taken over $x$ with respect to $\mu$.
\end{lem}
\begin{proof}
Since the empirical distribution $\mu_n$ converges to $\mu$, by the Portmanteau theorem \cite{van1996weak}, it is straightforward that
\begin{equation*}
    \int f(x) d\mu_n(x) \to \int f(x) d\mu(x) =\mathbb{E}_\mu(f(x))< \infty,
\end{equation*}
for continuous and bounded $f(x)$ on compact $\mathcal{Q}$.
Similarly, for any continuous and bounded function $f(x,c)$ on $\mathcal{Q} \times \mathcal{Q}^0$, the point-wise convergence holds, i.e. for every $c \in \mathcal{Q}^0$,
    \begin{equation*}
        \int f(x,c) d\mu_n(x) \to \int f(x,c) d\mu(x) = \mathbb{E}_\mu(f(x,c)), ~ n \to \infty.
    \end{equation*}
Since $f(x,c)$ is continuous on the compact set $\mathcal{Q} \times \mathcal{Q}^0$, it is uniformly continuous on $\mathcal{Q} \times \mathcal{Q}^0$. Thus, for every $\varepsilon > 0$, there exists $\delta_c > 0$ such that for all $x \in \mathcal{Q}$ and $c_1,c_2 \in \mathcal{Q}^0$, if $\|c_1-c_2\| < \delta_c$, then $\|f(x,c_1)-f(x,c_2)\|_F < \varepsilon$. Here, $\|\cdot\|_F$ represents the Frobenius norm of a matrix $(\cdot)$, and it coincidices with the 2-norm when $(\cdot)$ is a vector. Then
\begin{align*}
    &\left\|\int f(x,c_1) d\mu_n(x)- \int f(x,c_2) d\mu_n(x)\right\|_F\\
    &\hspace{15mm}\leq  \int \left\|f(x,c_1)-f(x,c_2)\right\|_F d\mu_n(x)
    \leq  \varepsilon.
\end{align*}
    Then $\{ g_n(c) = \int f(x,c) d\mu_n(x) \}_{n=1}^\infty$ is a family of equicontinuity functions. Thus, by the Arzelà–Ascoli Theorem \cite{arzele1895}, we conclude that 
    \begin{equation*}
        \int f(x,c) d\mu_n(x) \to \int f(x,c) d\mu(x) = \mathbb{E}_\mu(f(x,c)),~ n \to \infty
    \end{equation*}
    uniformly on $\mathcal{Q}^0$.
\end{proof}

\begin{lem}\cite[Theorem 14.4-1 on page 476]{Bishop2007}
	\label{lem_sqrtn}
	Let $\{X_k\}$ be a sequence of independent random variables with $\mathbb{E}X_k = 0$ and $\mathbb{E}X_k^2 \leq C < \infty$ for all $k$ and a positive constant $C$. Then, there holds that $\frac1n\sum_{k=1}^{n}X_k = O_p(1/\sqrt{n})$.
\end{lem}


\end{document}